%% file: AdmissibleStrategiesFull.tex
\documentclass[a4paper,11pt,envcountsame]{article}

\usepackage[english]{babel}
\usepackage{amsmath,amssymb,enumitem,xspace,tikz,stmaryrd,mathpartir,stackrel,wrapfig,capt-of}
\usepackage{subfig}
 
\usepackage[vlined, linesnumbered]{algorithm2e}
\usepackage{complexity}
\usepackage[hidelinks]{hyperref}
\usepackage{amsthm}
\theoremstyle{plain}
\newtheorem{definition}{Definition}[section]
\newtheorem{theorem}[definition]{Theorem}
\newtheorem{lemma}[definition]{Lemma}
\newtheorem{corollary}[definition]{Corollary}
\newtheorem{proposition}[definition]{Proposition}
\theoremstyle{remark}
\newtheorem{example}[definition]{Example}

\usetikzlibrary{petri,calc,automata,arrows,decorations.pathreplacing,decorations.pathmorphing,fit,positioning,patterns,shapes}

\newif\ifShort
\Shortfalse
\newif\ifDraft
\Draftfalse
\usepackage{fixme,manfnt,mparhack}
\newcounter{fixcount}
\newcommand{\defineNote}[3][black!65!green]{\expandafter\def\csname #2\endcsname ##1{\stepcounter{fixcount}\fxwarning{\textcolor{#1}{\textbf{#3}: ##1}}}}

\defineNote{RB}{Romain}
\defineNote{JFR}{Jean-Fran\c cois}
\defineNote{MS}{Mathieu}

\input{macrosElimstrat}

\title{The Complexity of Admissibility in Omega-Regular Games\thanks{Work supported by ERC Starting Grant inVEST (279499).}}
\author{Romain Brenguier \and Jean-Fran\c cois Raskin \and Mathieu Sassolas}
\date{D\'epartement d'informatique, Universit\'e Libre de Bruxelles (U.L.B), Belgium}

\newenvironment{collect}[3]{}{}
\newenvironment{collect*}[5]{}{}

\makeatletter
\usepackage{xstring}
\let\subp@r\subparagraph
\renewcommand{\subparagraph}[1]{%
\noexpandarg
\StrRight{#1}{1}[\endch@r]
\fullexpandarg
\IfEndWith{\endch@r}{.}{\subp@r{#1}}{\subp@r{#1.}}
}
\makeatother

\begin{document}

\maketitle

\input{abstractintro}

\subparagraph{Organization of the paper.}
The rest of the paper is organized as follows.
We first formalize the setting and the notations in Section~\ref{sec:background}.
Then we solve the case of games with safety objectives, in Section~\ref{sec:safety}, which gives rise to a simple notion of dominance.

We then give the more general algorithm and constructions for Muller objectives in Section~\ref{sec:prefix-independent}.
This include the special case of B\"uchi objectives (Section~\ref{sec:buchi}) as well as the algorithm solving the model-checking under admissibility problem (Section~\ref{sec:ltl}).

A comprehensive example of iterated elimination of dominated strategies and model-checking under admissibility is provided in Section~\ref{sec:train}.

In Section~\ref{sec:weak}, we generalize the algorithm for safety with weak Muller objectives.

\section{Definitions}\label{sec:background}

\input{definitions}

\section{Safety objectives}
\label{sec:safety}

\input{safety}

\section{Prefix-independent objectives}\label{sec:prefix-independent}

\input{transformation}
\input{automata}

\subsection{B\"uchi objectives}\label{sec:buchi}
\input{buchi}

\input{LTL}

\input{train}

\section{Weak objectives}\label{sec:weak}
\input{reachability}

\section{Conclusion}

\input{conclusionFull}
\label{sec:conclusion}

\bibliography{AdmissibleStrategies}
\end{document}

%% file: abstractintro.tex
\begin{abstract}
Iterated admissibility is a well-known and important concept in classical game theory, e.g. to determine {\em rational} behaviors in multi-player matrix games.
As recently shown by Berwanger, this concept can be soundly extended to infinite games played on graphs with $\omega$-regular objectives.
In this paper, we study the algorithmic properties of this concept for such games.
We settle the exact complexity of natural decision problems on the set of strategies that survive iterated elimination of dominated strategies.
As a byproduct of our construction, we obtain automata which recognize all the possible outcomes of such strategies.
\end{abstract}

\section{Introduction}
	
Two-player games played on graphs are central in many applications of computer science.  
For example, in the synthesis problem for reactive systems, implementations are obtained from winning strategies in games with a $\omega$-regular objectives~\cite{PnueliR89}.
To analyze systems composed of several components, two-player games are extended to multi-player games with non zero-sum objectives, \ie each player has his own objective expressed as a $\omega$-regular specification which is not necessarily adversarial w.r.t. the objectives of the other players. 

\begin{wrapfigure}{r}{0.4\columnwidth}
\begin{center}
\begin{math}
	\begin{array}{c | cc}
		    & C & D \\ \hline 
		 A & (0,2) & (1,1) \\
		 B & (1,1) & (1,2)
	\end{array}
\end{math}
\end{center}
\small
\caption{A two-player matrix game.}
\label{matrix}
\end{wrapfigure}
To analyze multi-player games in normal form (a.k.a. {\em matrix games}), concepts like the celebrated Nash equilibrium~\cite{nash50} have been proposed. Another central concept is the notion of {\em dominated strategy}~\cite{Rub94}. A strategy of a player {\em dominates} another one if the outcome of the first strategy is better than the outcome of the second no matter how the other players play. In two-player matrix game of \figurename~\ref{matrix}, strategies of \player{1} (of \player{2} respectively) are given as rows of the matrix (as columns respectively), and the payoffs to be maximized, are given as pairs of integers (the first for \player{1} and the second for \player{2}).  
Strategy $B$ of \player{1} dominates strategy $A$: no matter how \player{2} plays, $B$ provides an outcome which is larger than or equal to the one of $A$, and if \player{2} plays $C$ then the outcome provided by $B$ is strictly larger than the outcome of $A$. 
On the other hand, \player{2} at first sight has no preference between $C$ and $D$. But if \player{2} knows that \player{1} prefers strategy $B$ to strategy $A$, then he will in turn prefer $D$ to $C$, and it is then reasonable to predict that $(B,D)$ will be played. This process is called the {\em iterated elimination of dominated strategies}, and it is valid under the hypothesis that rationality is {\em common knowledge} among the players~\cite{Aumann76}. Strategies that survive the iterated elimination of strategies are called {\em iteratively admissible strategies}. 

In~\cite{berwanger07}, Berwanger initiated a fundamental study of the notion of {\em rational behaviour} in infinite duration games played on graph by generalising the notion of strategy dominance and iterated elimination of dominated strategies to that setting. This solution concept is a clear potential alternative to Nash equilibria for those games. As pointed out by Berwanger, one important advantage of admissible strategies is that they are compatible with the sequential nature of games played on graphs: {\it ``in any position reachable with an admissible strategy, a strategy is admissible in the sub-game rooted in that position if, and only if, it is the restriction of an admissible strategy in the original game.''} As a consequence, admissibility does not feature {\em non-credible threats} while it is well known that it is the case for Nash equilibria. Nonetheless, the extension of iterated strategy elimination to infinite duration games is challenging as the set of strategies is infinite and may lead to infinite dominance chains. Berwanger's main technical results are as follows: all iteration stages are dominated by admissible strategies, the iteration is non-stagnating, and, under regular objectives, admissible strategies form a {\em regular set}. In particular, for the last result, Berwanger suggests a procedure that uses tree automata to represent sets of strategies. The closure of tree automata to projection and Boolean operations naturally provides an algorithm to compute admissible strategies in parity games but this algorithm has {\em non-elementary complexity}. 

In order to represent a viable alternative to Nash equilibria from a {\em computational point of view}, it is fundamental to better understand the complexity of iterated elimination of dominated strategies in $\omega$-regular games, and see if the non-elementary complexity of the tree-automata based procedure can be avoided. We prove here that this is indeed the case and that iterated elimination of dominated strategies has a computational complexity similar to the one of Nash equilibria. More precisely, we study games with weak Muller and (classical) Muller winning conditions given as circuits. Circuits offer a concise representation of Muller conditions and are closed (while remaining succinct) under Boolean operations.
Weak Muller conditions define objectives based on the set of states that occur along a run, they generalize safety and reachability objectives.
(Classical) Muller conditions define objectives based on the set of states that appear \emph{infinitely often} along a run.
They generalize B\"uchi and parity objectives and are canonical representations of $\omega$-regular objectives as every regular language can be accepted by a deterministic Muller automaton. We study the {\em winning coalition problem}: given a game and two subsets $W,L$ of players, to determine whether there exists an iteratively admissible profile of strategies that guarantees that $(i)$ all players of $W$ win the game, and $(ii)$ all players of $L$ lose the game (other players may either win or lose). For weak and classical Muller objectives, we provide a procedure in \PSPACE, with matching lower-bounds for safety, reachability, and Muller objectives. For B\"uchi objectives, we provide an algorithm that calls a polynomial number of times an oracle solving parity games (hence this would lead to a polynomial time solution if a polynomial time algorithm is found for parity games -- the current best known complexity is $\UP \cap \coUP$~\cite{Jur98}, although a deterministic subexponential algorithm exists~\cite{jurdzinski06}).


As a byproduct of our constructions, we obtain an automaton on infinite words which recognizes {\em all the possible outcomes} of iteratively admissible strategies. Any regular query on this language can be solved using classical automata techniques. As a consequence, we can solve any variant of the winning coalition problem defined above, if this variant can be expresse as such a query. For example, we can solve the \emph{model-checking under admissibility} problem: given $\varphi$, an \textsf{LTL} formula~\cite{pnueli77,sistla85}, does the outcome of every iteratively admissible profile satisfy $\varphi$?
We show that this problem is complete for the class \PSPACE, so it retains the same complexity as the ``classical'' model-checking problem for this logic. Model-checking under admissibility is useful to reason about properties that naturally {\em emerge} in a system from the interaction of {\em rational} agents that purse {\em their own objectives}. 

\subparagraph{Related work}
Dominance can be expressed in strategy logics~\cite{CHP10,mogavero10} but not unbounded iterated dominance.
Bounded iterated dominance is expressible but leads to classes of formulas with a non-elementary model-checking algorithm.
Other paradigms of rationality have been studied for games on finite graphs, like Nash-equilibria~\cite{Ummels08,bouyer12,klimos12} or \emph{regret minimization}~\cite{filiot10}.
In~\cite{klimos12}, the authors build an automaton that recognizes outcomes of Nash equilibria.
In turn-based game, finding a Nash equilibrium with a particular payoff is \PSPACE-complete for Muller objectives~\cite{Ummels10}, which is the same complexity we obtain for admissibility.
In the case of B\"uchi objectives, a polynomial algorithm exists for Nash equilibria~\cite{Ummels08}, while we have a $\NP \cap \co\NP$~algorithm for admissibility.

In this paper, we concentrate on $n$-player turn-based perfect information finite game graph with $\omega$-regular objective. This is the basic setting that needs to be studied before looking at richer models, like games with incomplete information~\cite{Reif84}, games with quantitative objectives like mean-payoff objectives~\cite{ZwickP96}, or concurrent games~\cite{AlfaroHK07}. Our results and techniques are clearly prerequisites to study those richer settings.

%% file: definitions.tex
\subsection{Multiplayer Games}



\begin{definition}[Multiplayer games] 
  A \newdef{turn-based multiplayer (non zero-sum) game} is a tuple
  $\G = \left\langle \Agt, (V_i)_{i\in\Agt}, E, (\win{i})_{i\in\Agt} \right\rangle$
  where:
\begin{itemize}
\item $\Agt$ is the non-empty and finite set of players;
\item 
  $V = \biguplus_{i\in \Agt} V_i$ and for every $i$ in $\Agt$, $V_i$ is the finite set of player $i$'s states; 
\item $E \subseteq V \times V$
  is the set of edges\footnote{It is assumed for convinience and w.l.o.g. that each state in $V$ has at least one outgoing edge.};
  we write $s\rightarrow t$ for $(s,t)\in E$ when $E$ is clear from context.
\item For every $i$ in $\Agt$, $\win{i} \subseteq V^\omega$ is a
  winning condition. 
\end{itemize}
\end{definition}

A \newdef{path}~$\rho$ is a sequence of states $(\rho_j)_{0\le j < n}$ with $n\in \N \cup \{\infty\}$ s.t. for all $j <n -1$, $\rho_j \rightarrow \rho_{j+1}$.
The \newdef{length}~$|\rho|$ of the path $\rho$ is $n$.\MS{Je crois qu'on avait une typo ici; ou alors on compte le nombre de transitions travers\'ees?}
A \newdef{history} is a finite path and a \newdef{run} is an infinite path.
Given a run~$\rho= (\rho_j)_{j\in\N}$ and an integer $k$, we write $\rho_{\le k}$ the history $(\rho_j)_{0\le j< k+1}$, that is, the prefix of length $k+1$ of~$\rho$.
For a history $\rho$ and a (finite or infinite) path $\rho'$, $\rho$ is a \newdef{prefix} of $\rho'$ is written $\rho \prefix \rho'$.\MS{La macro $\preccurlyeq$ \'etant prise pour la dominance, j'ai choisi \c ca un
  peu au pif; \`a changer, \'eventuellement en faisant un gros shuffle
  des macros :-)}
The \newdef{last state} of a history $\rho$ is $\last(\rho) = \rho_{|\rho|-1}$.
The set of states \newdef{occuring} in a path~$\rho$ is $\states{\rho} = \left\{ s \in V \mmid \exists i \in \N.\ i < |\rho|, \rho_i = s \right\}$.
The set of states \newdef{occuring infinitely often} in a run~$\rho$ is $\inff{\rho} = \left\{ s \in V \mmid \forall j\in \N.\ \exists i > j, \rho_i = s \right\}$.


\begin{definition}[Strategies]
  A \newdef{strategy} of player $i$ is a function 
  $\sigma_i : (V^* \cdot V_i) \rightarrow V$, such that if $\sigma_i(\rho) = s$ then $(\last(\rho),s)\in E$.
  A \newdef{strategy profile} for the set of players~$P' \subseteq \Agt$ is a tuple of
  strategies, one for each player of~$P'$.
\end{definition}

Let $\stratset_{i}(\G)$ be the set of all strategies of player $i$ in $\G$; we write $\stratset_i$ when $\G$ is clear from the context.
We write $\stratset = \prod_{i\in\Agt}\stratset_i$ for the set of all strategy profiles for $\Agt$, and $\stratset_{-i}$ for the set of strategy profiles for all players
but $i$.
If $\sigma_{-i} = (\sigma_j)_{j\in \Agt\setminus\{i\}} \in \stratset_{-i}$, we write $(\sigma_i,\sigma_{-i})$ for $(\sigma_j)_{j\in \Agt}$.
Similarly, if $S$ is a set of profiles, $S_i$ denotes the $i$-th projection of $S$, \ie a set of strategies for \player{i}.
A \emph{rectangular set} of strategy profiles is a set that can be decomposed as a Cartesian product of sets of strategies, one for each player.

A strategy profile $\sigma_\Agt \in \stratset$ defines a unique \newdef{outcome} from state~$s$:
$\outcome_s(\sigma_\Agt)$ is the run $\rho = (\rho_j)_{j\in \N}$ s.t. $\rho_0 = s$ and for $j \ge 0$, if $\rho_j \in V_i$, 
then $\rho_{j+1} = \sigma_i(\rho_{\le j})$. 
If $S_i$ is a set of strategies for \player{i}, we write $\outcome_s(S_i)$ for $\{\rho \mid \exists \sigma_i \in S_i , \sigma_{-i}\in \stratset_{-i}.\ \outcome_s(\sigma_i,\sigma_{-i}) = \rho \}$.
For a tuple of sets of strategies $S_{P'}$ with $P' \subseteq \Agt$, we write $\outcome_s(S_{P'}) = \bigcap_{i\in P'} \outcome_s(S_i)$.
A strategy $\sigma_i$ of player $i$ is said to be \newdef{winning from
state~$s$ against a rectangular set $S_{-i} \subseteq \stratset_{-i}$}, if for all $\sigma_{-i}\in S_{-i}$, $\outcome_s(\sigma_i,\sigma_{-i}) \in \win{i}$.
It is simply said \newdef{winning from state~$s$} if $S_{-i} = \stratset_{-i}$.
For each player $i$, we write $\win{i}^s(\sigma_\Agt)$ 
if $\outcome_s(\sigma_\Agt) \in \win{i}$.


\subsection{Winning conditions}
Winning conditions for each player are given by accepting sets
either on the set of states occurring along the run, or the set of states occurring infinitely often.
Particular cases are safety, reachability, and B\"uchi winning conditions.

\begin{itemize}
\item
  A \newdef{safety condition} is defined by a set $\bad_i \subseteq V$:
 $\win{i} = (V\setminus \bad_i)^\omega$.
\item
  A \newdef{reachability condition} is 
  defined by a set $\good_i \subseteq V$:
  $\win{i} = V^* \cdot \good_i \cdot V^\omega$.
\item
  A \newdef{B\"uchi condition} is defined by a set $\buchi_i \subseteq V$:
  $\win{i} = (V^*\cdot\buchi_i)^\omega$.
\item
  A \newdef{Muller condition} is given by a family~$\mathcal{F}$ of sets of states: $\win{i} = \{ \rho \mid \inff{\rho}\in \mathcal{F}\}$.
  For a succinct representation, we assume $\mathcal{F}$ is given by a Boolean circuit whose inputs are the states of $V$, and which evaluates to true on valuation $v_S \colon s \mapsto 1 \textrm{ if } s \in S ; \ 0 \textrm{ otherwise};$
  if, and only if, $S\in \mathcal{F}$~\cite{hunter07}.
\item
  A \newdef{weak Muller condition} is given by a family~$\mathcal{F}$ of sets of states: $\win{i} = \{ \rho \mid \states{\rho}\in \mathcal{F}\}$.
  We again assume that $\mathcal{F}$ is given by a Boolean circuit.
\end{itemize}

Muller conditions generalize B\"uchi and other classical conditions such as parity: these can be encoded by a circuit of polynomial size~\cite{hunter07}.
Note that Muller conditions are \newdef{prefix-independent}: for any finite path $h$ and
infinite path $\rho'$, $h\cdot \rho' \in \win{i} \Leftrightarrow \rho' \in \win{i}$.
In two-player zero-sum games with circuit conditions, deciding the winner is \PSPACE-complete~\cite{hunter07}.
Weak Muller conditions generalize safety and reachability.

\subsection{Admissibility}

\begin{definition}[Dominance for strategies]
Let $S = \prod_{i \in \Agt}S_i \subseteq \stratset$ be a rectangular set of strategy profiles.\JFR{Here there is a game which is implicite, we should say it}
Let $\sigma,\sigma' \in S_i$.
Strategy $\sigma$ \newdef{very weakly dominates} strategy $\sigma'$ with
respect to $S$, written $\sigma \dom{S} \sigma'$, if from all states~$s$:
\[\forall \tau \in S_{-i}, \win{i}^s(\sigma',\tau) \Rightarrow \win{i}^s(\sigma,\tau).\]

Strategy $\sigma$ \newdef{weakly dominates} strategy $\sigma'$
with respect to $S$, written $\sigma \domstr{S} \sigma'$, if
$\sigma \dom{S} \sigma'$ and $\neg(\sigma' \dom{S} \sigma)$. 
A strategy $\sigma \in S_i$ is \newdef{dominated} in $S$ if there exists $\sigma' \in S_i$ such that $\sigma' \domstr{S} \sigma$.
A strategy that is not dominated in $S$ is \newdef{admissible} in $S$.
A profile $\sigma_\Agt$ such that for all $i \in \Agt$, $\sigma_i$ is admissible is called an \newdef{admissible profile}.
\end{definition}

The set~$\stratset^*$ of \newdef{iteratively admissible} strategies is obtained by iteratively eliminating dominated strategies, starting from set $\stratset$.
Formally, we consider the sequence: 
\begin{itemize}
\item $\stratset^0 = \stratset$;
\item $\stratset^{n+1} = \prod_{i\in\Agt} \{ \sigma_i \in \stratset^n_i \mid \sigma_i\text{ admissible in } \stratset^n\}$.
\end{itemize}
Then $\stratset^* = \bigcap_{n\in \N} \stratset^n$.
When for all player~$i\in \Agt$, $\win{i}$ is $\omega$-regular winning conditions, 
$\stratset^*$ is reached after finitely many iterations and is not empty~\cite{berwanger07}.

Note that our strategies are defined from all states while in~\cite{berwanger07} they are defined only for history starting from the initial state.
A strategy here can be seen as a tuple of strategies (in the sense of~\cite{berwanger07}): one for each state.
The set $\stratset^n$ we compute is then the cardinal product of admissible strategies from each state.

\begin{example}
  \figurename~\ref{fig:exsafety} presents a safety game that starts in $q_0$.
  Strategies of \player{1} that from $q_0$ go to $q_4$ are losing.
  Whereas for those that go to $q_1$, there is a strategy of \player{2} which helps \player{1} to win by playing back to $q_0$.
  Hence the former are dominated by the later, and so they are eliminated at the first elimination of dominated strategies and do not appear in $\stratset^1_1$.
  In the second step of iteration, if \player{2} plays from $q_1$ to $q_0$, he is ensured to win if \player{1} plays a strategy of $\stratset^1_1$.
  Therefore the strategies of \player{2} that go to $q_2$ are dominated: there are strategies in $\stratset^1_1$ that make them lose.
  These latter strategies are therefore removed and do not appear in $\stratset^2_2$.
  The process then stabilizes: $\stratset^*=\stratset^2$.
\end{example}

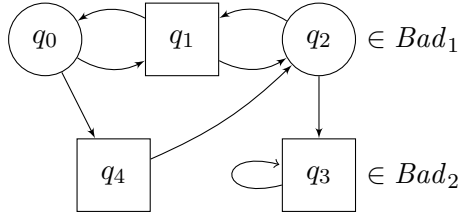
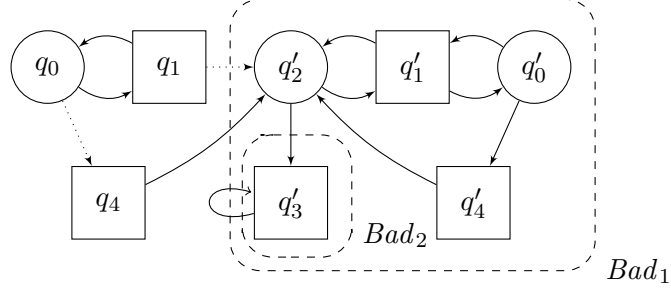
\begin{figure}[t]
\begin{center}
\subfloat[A safety game.]{\label{fig:exsafety}
  \begin{tikzpicture}[xscale=1.8,yscale=1.8]
    \draw (0,0) node[player1] (A) {$q_0$};
    \draw (1,0) node[player2] (B) {$q_1$};
    \draw (2,0) node[player1] (C) {$q_2$};
    \draw (C.0) node[right] {$\in \bad_1$};
    \draw (2,-1) node[player2] (D) {$q_3$};
    \draw (D.0) node[right] {$\in \bad_2$};
    \draw (0.5,-1) node[player2] (E) {$q_4$};
    
    \draw[-latex'] (A) -- (E);
    \draw[-latex'] (C) -- (D);
    \draw[-latex'] (E) edge[bend right=10] (C);

    \draw (A) edge[bend right,-latex'] (B);
    \draw (B) edge[bend right,-latex'] (A);
    \draw (B) edge[bend right,-latex'] (C);
    \draw (C) edge[bend right,-latex'] (B);
    \draw (D) edge[loop left,-latex'] (D);
  \end{tikzpicture}
}

\subfloat[Its unfolding. Transitions eliminated after two steps of elimination are dotted.\vspace{-1em}]{\label{fig:after-elim}
  \begin{tikzpicture}[xscale=1.6,yscale=1.8]
  \tikzstyle{dominated}=[dotted]
    \draw (0,0) node[player1] (A) {$q_0'$};
    \draw (-1,0) node[player2] (B) {$q_1'$};
    \draw (-2,0) node[player1] (C) {$q_2'$};
    \draw (-2,-1) node[player2] (D) {$q_3'$};
    \draw (-0.5,-1) node[player2] (E) {$q_4'$};
    
    \draw (-3,0) node[player2] (Bp) {$q_1$};
    \draw (-4,0) node[player1] (Ap) {$q_0$};
    \draw (-3.5,-1) node[player2] (Ep) {$q_4$};

    \draw[-latex',dominated] (Ap) -- (Ep); 
    \draw[-latex',dominated] (Bp) -- (C); 
    \draw[-latex'] (A) -- (E);
    \draw[-latex'] (C) -- (D);
    \draw[-latex'] (E)  edge[bend left=10] (C);
    \draw[-latex'] (Ep)  edge[bend right=10] (C);

    \draw (Bp) edge[bend right,-latex'] (Ap);
    \draw (Ap) edge[bend right,-latex'] (Bp);
    \draw (B) edge[bend right,-latex'] (A);
    \draw (B) edge[bend right,-latex'] (C);
    \draw (A) edge[bend right,-latex'] (B);
    \draw (C) edge[bend right,-latex'] (B);
    \draw (D) edge[loop left,-latex'] (D);
    \draw[rounded corners=4mm,dashed] (0.5,0) |- (-1,0.5) -| (-2.5,-1) |- (-1,-1.5) -| node[right]{$\bad_1$} (0.5,0);
    \draw[rounded corners=4mm,dashed] (-2,-1.4) -|node[above right]{$\bad_2$}  (-1.5,-1) |- (-2,-0.5) -| (-2.4,-1) |- (-2,-1.4);
  \end{tikzpicture}
}
\end{center}
\caption[A safety game and its unfolding.]{A safety game and its unfolding. States controlled by \player{1} are represented with circles and states controlled by \player{2} with squares. We keep this convention through the paper.
  The objective of \player{1} is to avoid $q_2$ and that of \player{2} is to avoid $q_3$.
}
\end{figure}

\subsection{Decision problems}
\subparagraph{Winning coalition problem} 
Given a game~$\G$ and two subsets $W,L$ of players, 
does there exist an iteratively admissible profile s.t. all players of $W$ win the game, and all players of $L$ lose the game (other players may either win or lose)?
\JFR{what about LTL}

\subparagraph{Model-checking under admissibility problem}\MS{Le nom est un peu au pif :-)}
Given a game and an \textsf{LTL}~\cite{pnueli77,sistla85} formula $\psi$,
does the outcome of every iteratively admissible profile satisfy $\psi$?

\subsection{Values}
Our algorithms are based on the notion of \emph{value} of a history.
It characterizes whether a player can win (alone) or cannot win (even with the help of other players), restricting the strategies to the ones that have not been eliminated so far.
This notion is also a central tool in~\cite{berwanger07} to characterize admissible strategies. 
However,~\cite{berwanger07} gives no practical way to compute values.
We will show in this paper, that these are indeed computable.

\begin{definition}[Value]
  The \newdef{value} of history~$h$ for \player{i} after
  the $n$-th step of elimination, written $\val^{n}_i(h)$, is given by:
  \begin{itemize}
  \item if there is no strategy profile $\sigma_{\Agt}$ in
    $\stratset^n$ whose outcome $\rho$ from $\last(h)$ is such that $h_{< |h| - 1} \cdot \rho$ 
    is winning for \player{i} then
    $\val^{n}_i(h)=-1$;
  \item if there is a strategy of $\sigma_i \in \stratset^n_i$ such that 
    for all strategy profiles $\sigma_{-i}$ in
    $\stratset^n_{-i}$, the outcome $\rho$ of $(\sigma_i,\sigma_{-i})$ from
    $s$ is such that $h_{< |h| - 1}\cdot\rho$ is winning for \player{i} then $\val^{n}_i(h)=1$;
  \item otherwise $\val^{n}_i(h)=0$;
  \end{itemize}
  By convention, $\val^{-1}_i(h)=0$.
\end{definition}

\noindent
The following lemma illustrates a property of values and admissible strategies:
\begin{lemma}\label{lem:necessaryAdm}
For all $n\in\N$, if $\last(h) \in V_i$ and $\sigma_i \in \stratset^{n+1}_i$ then $\val^n_i(h) = \val^n_i(h \cdot \sigma_i(h))$.
\end{lemma}
Hence a player that plays according to an admissible strategy cannot go to a state that changes the value of the current history.
This condition is not always sufficient, but in the following sections we characterize runs of admissible strategies relying on this notion of value.

%% file: safety.tex

The main result of this section is a \PSPACE~algorithm for the winning coalition problem in safety games.
This is based on a notion of dominance for transitions.
We show that by iteratively removing dominated transitions of the game, we describe exactly the set of admissible strategies.

\subsection{Making explicit the losing players}

Let $h$ be an history, the players \newdef{losing on $h$} are the players in \[\lost(h)= \{i \in \Agt \mid \exists k < |h|,\ h_k \in \bad_i \}.\]

\begin{proposition}\label{prop:locallost}
  For safety winning conditions, the value of a history $h$ only depends on $\lost(h)$ and $\last(h)$.
\end{proposition}
\begin{collect}{appendix-safety}{\subsubsection{Proof of Proposition~\ref{prop:locallost}}}{}
\begin{proof}
  Since a safety objective can be transformed into a prefix-independent one by remembering which player have already lost, this is a consequence of the fact that for prefix-independent objectives the value depends only on the last state of the history, as we will see in Proposition~\ref{prop:valuestate}.
\end{proof}
\end{collect}
We can therefore write $\val^n_i(\lost(h),s)$\JFR{this \dots} for $\val^n_i(h)$, when $\last(h) = s$.
We encode the set $\lost(h)$ of losing players in the state of the game, at the price of an exponential blowup (in the number of players).
The new game has states in $2^{\Agt} \times V$ and set of transitions $(\lost, s) \rightarrow (\lost\cup \{i \mid s' \in \bad_i\} , s')$ for any $\lost\subseteq \Agt$, if $s \rightarrow s'$.
In this partially unfolded game, the value depends only on the current state, hence is written $\val^n_i(s)$.
For example, the game of \figurename~\ref{fig:exsafety} is unfolded as the game of \figurename~\ref{fig:after-elim}; states $q_0', \ldots, q_4'$ are states where \player{1} has already lost.
Now, let us assume for the remainder of this section that the losing players in the game $\G$ are explicit.

\subsection{Dominance of transitions}

In the case of safety winning condition, the necessary condition of Lemma~\ref{lem:necessaryAdm} becomes sufficient, as shown below.\JFR{explain here what is this notion and intuition}
This yields a local notion of dominance, that can be expressed directly on transitions:

\begin{definition}
  We write $T_i^n$ for the set of transitions $s \rightarrow s' \in E$, such that $s$ is controlled by \player{i} and $\val^n_i(s) > \val^n_i(s')$.
  Such transitions are said to be \newdef{dominated} after the $n$-th step of elimination.
  We write $T^n$ for the union of all $T_i^n$.
\end{definition}

\begin{definition}[Subgame]
  Let $\G = \left\langle \Agt,V,E,\win{\Agt}\right\rangle$ be a game and $T \subseteq E$ a set of transitions.
  If each state $s\in V$ has at least one successor by $E\setminus T$, the game $\G\setminus T = \left\langle \Agt,V,E\setminus T,\win{\Agt}\right\rangle$ is called a \newdef{subgame} of~$\G$.
We write $\stratset_i(\G\setminus T)$ the set of strategies $\sigma_i\in \stratset_i(\G)$ such that for all history $h$ of $\G\setminus T$, if $\last(h)\in V_i$ then $(\last(h),\sigma_i(h)) \not\in T$.
\end{definition}


This 
notion yields a polynomial procedure in the size of the game where losing players are explicit, to compute the set of all iteratively admissible strategies, described in Algorithm~\ref{algo:safety}.
The loop is executed at most $|E|$ times, where $|E|$ is the number of transitions in the partially unfolded game.
\begin{algorithm}[htb]
  $n := 0$ ; $T_i^{-1} := \emptyset$ \;
      \Repeat{$\forall i \in \Agt.\ T_i^{n} = T_i^{n-1}$}
      {
        \ForAll{$s \in V$}
               {
               \lIf{there is a winning strategy for \player{i} from $s$ in $\G \setminus T^{n-1}$}
                     {$\val^n_i(s) := 1$}
                     {
                       \lElseIf{there is no winning run for \player{i} from $s$ in  $\G \setminus T^{n-1}$}
                              {$\val^n_i(s) := -1$}
                              \lElse {$\val^n_i(s) := 0$}
                     }
               }
       \ForAll{$i\in \Agt$}
       {$T_i^{n} := T_i^{n-1} \cup \{ (s, s') \in E \mid s\in V_i \land \val^n_i(s) > \val^n_i(s')\}$\;
       }
        $n := n+1$ \;
      }       

      \caption{Computing the set of iteratively admissible strategies}
      \label{algo:safety}
\end{algorithm}

However, this procedure assumes that the information of which players have already violated their safety condition is encoded in the state.
So in the general case, the procedure has a complexity which is exponential in the number of players and polynomial in the number of states of the game.
In the case of the \emph{winning coalition problem}, we can however reduce this complexity to $\PSPACE$.

We now show the correctness of the procedure.
We first prove a link between the notions of dominance for strategies and for transitions.
Note that since all states~$s$ have at least one successor with a value greater or equal to that of $s$, removing transitions of $T_i^n$ yield what we call a subgame.

\begin{proposition}\label{prop:correctness}
  All admissible strategies w.r.t. $\stratset^n$ of \player{i} are strategies of $\stratset_i(\G \setminus T_i^n)$.
\end{proposition}
\begin{proof}
  We show that if \player{i} plays an strategy~$\sigma_i$ admissible w.r.t. $\stratset^n$, \ie $\sigma_i \in \stratset^{n+1}_i$, then the value cannot decrease on a transition controled by \player{i}.
  Let $\rho \in \outcome(\sigma_i,\sigma_{-i})$ with $\sigma_{-i} \in \stratset^n_{-i}$ and $\sigma_i \in \stratset_i^{n+1}$, and $\rho_k \in V_i$.
  Let $s' = \sigma_i(\rho_{\le k})$:
  \begin{itemize}
  \item If $\val^n_i(\rho_k)=1$, then $\sigma_i$ has to be a winning against all strategy of $\stratset^n_{-i}$, otherwise it would be weakly dominated by such a strategy.
    Since there is no such strategy from a state with value $\val^n_i \leq 0$, $\val^n_i(s')=1$.
  \item If $\val^n_i(s)=0$, then there is a profile $\sigma_{-i} \in \stratset^n_{-i}$ such that $\rho = \outcome(\sigma_i,\sigma_{-i}) \in \win{i}$.
    Note that $h \cdot s \cdot s' \prefix \rho$.
    If $\val^n_i(s')=-1$, there can be no such profile, thus $\val^n_i(s')\ge 0$.
  \item If $\val^n_i(s)=-1$, the value cannot decrease.\qedhere
  \end{itemize}
\end{proof}

\begin{example}
  In \figurename~\ref{fig:exsafety}, initially, $q_4$ has value $-1$ for \player{1}, but $q_0$ has value~$0$ since it is possible to loop in $q_1$ and $q_0$ (if \player{2} helps).
  So, the transition to state $q_4$ is dominated and removed at the first iteration.
  Then, \player{2} has a winning strategy from $q_1$, by always going back to $q_0$, whereas the state $q_2'$ has value $0$ for him.
  Hence $q_1 \rightarrow q_2'$ is removed after this iteration.
  The fix-point is obtained at that step, it is represented in \figurename~\ref{fig:after-elim}.
\end{example}

We have seen that removing dominated transitions only removes strictly dominated strategies.
The converse is also true, all strategies that remain are not dominated:

\begin{proposition}\label{prop:completeness}
  All strategies of $\stratset^n_i \cap \stratset_i(\G \setminus T_i^n)$ are admissible with respect to $\stratset^n$.
\end{proposition}

\begin{proof}
\RB{j'ai chang\'e pas mal de choses dans la preuve}
Let $\sigma_i,\sigma_i' \in \stratset^n_i \cap \stratset_i(\G \setminus T_i^n)$\RB{on devrait mettre $\sigma_i$ pour faire comme les sections d'apres} and assume $\sigma_i' \domstr{\stratset^n} \sigma_i$.
Then there is a state~$s$ and strategy profile $\sigma_{-i} \in \stratset_{-i}^n$ such that $\win{i}^s(\sigma_i',\sigma_{-i}) \wedge \neg \win{i}^s(\sigma_i,\sigma_{-i})$.
Let $\rho = \outcome_s(\sigma_i,\sigma_{-i})$ and $\rho' = \outcome_s(\sigma_i',\sigma_{-i})$.
Consider the first position where these runs differ: write $\rho = w \cdot s' \cdot s_2 \cdot w'$ and $\rho' = w \cdot s' \cdot s_1 \cdot w''$.
Note that $s'$ belongs to \player{i}.

First remark that since $\win{i}(\sigma_i',\sigma_{-i})$, it is clear that $\val_i^n(s_1) \geq 0$.
Moreover, since $s' \rightarrow s_1$ and $s\rightarrow s_2$ do not belong to $T^n_i$, states $s'$, $s_1$ and $s_2$ must have the same value.

\smallskip
Assume $\val_i^n(s') = 0$.
  We show that there is a profile\footnote{Although the definition of the value yields the existence of a profile winning for $i$, it remains to be shown that there is such profile where $i$ plays strategy $\sigma_i$.} $\sigma_{-i}^2 \in \stratset^n_{-i}$ such that $\win{i}(\sigma_i,\sigma_{-i}^2)$ from $s_2$.
  %
  Let $h$ be a history such that $\last(h)\notin V_i$, if for all $\sigma^2_{-i} \in \stratset^n_{-i}$, $\val_i^n(\sigma^2_{-i}(h)) = -1$ then $\val_i^n(\last(h))=-1$.
  Therefore it is possible to define a strategy profile $\sigma_{-i}^2\in\stratset^n_{-i}$ that never decreases the value from $0$ or $1$ to $-1$.
  The strategy $\sigma_i$ itself does not decrease the value of \player{i} because it does not take transitions of $T^n_i$.
  So the outcome of $(\sigma_i,\sigma_{-i}^2)$ never reaches a state of value $-1$.
  Hence it never reaches a state in $\bad_i$ and therefore it is winning for \player{i}.
Now, $\val_i^n(s_1) = 0$ so there is no winning strategy for \player{i} from $s_1$ against all strategies of $\stratset^n_{-i}$.
Then there exists a strategy profile $\sigma_{-i}^1 \in \stratset^n_{-i}$ such that $\sigma_i'$ loses from $s_1$.
Now consider strategy profile $\sigma_{-i}'$ that plays like $\sigma_{-i}$ if the play does not start with $w$, then $\sigma_{-i}^1$ after $s_1$ and $\sigma_{-i}^2$ after $s_2$.
Given a history $h$:
\[
\sigma_{-i}'(h) = \left\{\begin{array}{l l}
\sigma_{-i}^1( h') ~ ~ & \textrm{if } w \cdot s_1 \prefix h \text{ and } w \cdot s_1 \cdot h' = h\\
\sigma_{-i}^2( h') & \textrm{if } w \cdot s_2 \prefix h \text{ and } w \cdot s_2 \cdot h' = h\\
\sigma_{-i}(h) & \textrm{otherwise}
\end{array}\right.
\]
Clearly we have $\win{i}^s(\sigma_i,\sigma_{-i}') \wedge \neg \win{i}^s(\sigma_i',\sigma_{-i}')$, which contradicts $\sigma_i' \dom{\stratset^n} \sigma_i$.

\smallskip

  Now assume $\val_i^n(s_2) = 1$.
  Since $\neg \win{i}(\sigma_i,\sigma_{-i})$, the produced outcome~$\rho$ reaches a state of $\bad_{i}$, hence the value of states along~$\rho$ is $-1$ after some point.
  Consider the first state $\rho_k$ which has value smaller or equal to $0$: $k = \min_{k'} \{ \rho_{k'} \mid \val_i^n(\rho_{k'}) \le 0 \}$.
  The state $\rho_{k-1}$ has value $1$, it is necessarily controlled by a \player{j} different from \player{i}, since transitions of $T^n_i$ cannot be taken by $\sigma_i$.
  Since there exists a winning strategy $\sigma_i \in \stratset^n_i$ from $\rho_{k-1}$ against strategies of $\stratset^n_{-i}$, then this strategy is still winning at $\rho_k$.
  Therefore $\val^n_i(\rho_k)=1$, which is a contradiction.
\end{proof}


\subsection{The winning coalition problem for safety objectives}

\begin{theorem}\label{th:safety-pspace}
The winning coalition problem with safety winning conditions is \PSPACE-complete.
However, if the number of players is fixed, the problem becomes \P-complete.
\end{theorem}

\begin{collect}{appendix-safety}{\subsection{Proof of Theorem~\ref{th:safety-pspace}}}{}
This theorem is proved in the following two lemmata.

\begin{proposition}\label{lem:pspace-easy}
The winning coalition problem is in \PSPACE.
\end{proposition}

\begin{proof}
First remark that although there is an exponential number of copies of the game over the graph $(V,E)$ that need to be considered with respect to which players have already lost, states can be ordered the following way: we say that $(\lost,s) \le (\lost',s')$ if $\lost\subseteq \lost'$.
Along any path the states are increasing for this order, it can increase strictly at most $|\Agt|$ times, and there are at most $|V|$ equivalent states.
In addition, the value, hence the elimination of transitions, only depends on the values of greater states, so the iterations stops after at most $|\Agt|\cdot|E|$ phases.

Therefore a procedure to find an iteratively admissible strategy winning at least for players of $W$ and losing at least for players of $L$ consists in guessing a lasso path $\rho$ that ends in a copy where $W$ has not lost and $L$ has.
This path has length bounded by $|\Agt| \cdot |\G|$.

However the algorithm needs to check that each transition taken by $\rho$ has indeed survived the elimination of transitions: this transition should not be dominated by any other.
This is done by recursively checking that a transition has survived the $j$-th elimination phase (recall that there can be at most $|\Agt|\cdot|E|$ such phases).
For a transition $\rho_k \rightarrow \rho_{k+1}$ to survive the $j$-th phase, the value of $\rho_{k+1}$ needs to be the same than that of $\rho_k$ for the player controlling~$s$.
 
To check $\val^n_i(\rho_k)$, we use the following procedure:
\begin{itemize}
\item if we fail to guess a lasso which does not intersect with $\bad_i$ in $\G\setminus T^n$ from state $\rho_k$, then $\val^n_i(\rho_k) = -1$.
  Note that looking for a path in $\G\setminus T^n$ implies recursively computing some values of iteration $j-1$;
\item if there is a winning strategy for \player{i} in the safety game $\G\setminus T^n$ with target $\bad_i$,\RB{comment on fait \c{c}a en \PSPACE?}
  then $\val^n_i(\rho_k) = 1$; note that this can be done by finding a strategy that either never visits a new set $\bad_j$ (hence not increasing for $\le$) or visiting a new set $\bad_j$ through a state of value $1$ for $i$ (this value being computed recursively, for details, see the more general proof of Theorem~\ref{thm:weakobj}).
\item in the other cases $\val^n_i(\rho_k) = 0$.
\end{itemize}
In all cases, the recursive calls can stack up to $|\Agt| \cdot |\G|$\RB{a verifier}, since they always traverse the set of states upwards (with respect to $\le$).
\qedhere
\end{proof}

\begin{proposition}\label{lem:pspace-hard}
  The winning coalition problem is \PSPACE-hard, even for sets of players $W,L \subseteq \Agt$ such that $|W|=1$ and $L=\emptyset$.
\end{proposition}
\end{collect}

\begin{collect}{appendix-safety}{\subsubsection{Proof of Proposition~\ref{lem:pspace-hard}}}{}
\begin{proof}
We encode a instance of \QSAT  into a game in which there is an
admissible strategy profile which is winning for \Eve if, and only if,
the formula is valid.

Given a formula~$\phi = \exists x_1 \forall x_2 \dots \psi$
we associate a game $\G_\phi$ 
in which there are one player for each literal $x_i$ or $\lnot x_i$
and two players \Eve and \Adam.
The construction is recursive separately over the quantifiers and over the propositional part.
If $\psi$ is a propositional formula:
\begin{itemize}
\item If $\psi = x_i$ then we define the module $M_\psi$ in which
player $x_i$ has a choice between making \Eve lose or lose himself and
let the game continue, this is represented in \figurename~\ref{fig:Mx}.
\item If $\psi = \lnot x_i$ then the construction is similar, with
player $\lnot x_i$ replacing $x_i$, see \figurename~\ref{fig:Mnx}.
\item If $\psi = \psi_1 \land \psi_2$ then we put the modules
  $M_{\psi_1}$ and $M_{\psi_2}$ in sequence, see \figurename~\ref{fig:Mand}.
\item If $\psi = \bigvee_i \psi_i$ then \Eve has the choice between
  all modules $M_{\psi_i}$, see \figurename~\ref{fig:Mor}.
\end{itemize}
If $\phi$ is a quantified formula:
\begin{itemize}
\item If $\phi = \exists x_i.\ \phi_1$ then \Eve has the choice
  between making $x_i$ or $\lnot x_i$ lose before continuing to
  $M_{\phi_1}$, see \figurename~\ref{fig:Mexists}. 
\item If $\phi = \forall x_i.\ \phi_1$ is similar but \Adam controls
  the choice, see \figurename~\ref{fig:Mforall}. 
\end{itemize}
Finally $\G_\phi$ is obtained by directing the remaining outgoing transitions of $M_\phi$ to a state losing for \Adam, see \figurename~\ref{fig:G}.
A full example of $\G_{\mu}$ with $\mu = \exists x_1 \forall x_2 \exists x_3 (x_1 \lor x_2 \lor \lnot x_3) \land (\lnot x_1 \lor x_2 \lor x_3)$ is given in \figurename~\ref{fig:exsafetyhardness}.
Note that any run in $\G_\phi$ winning for \Eve is losing for \Adam, and \emph{vice versa}.

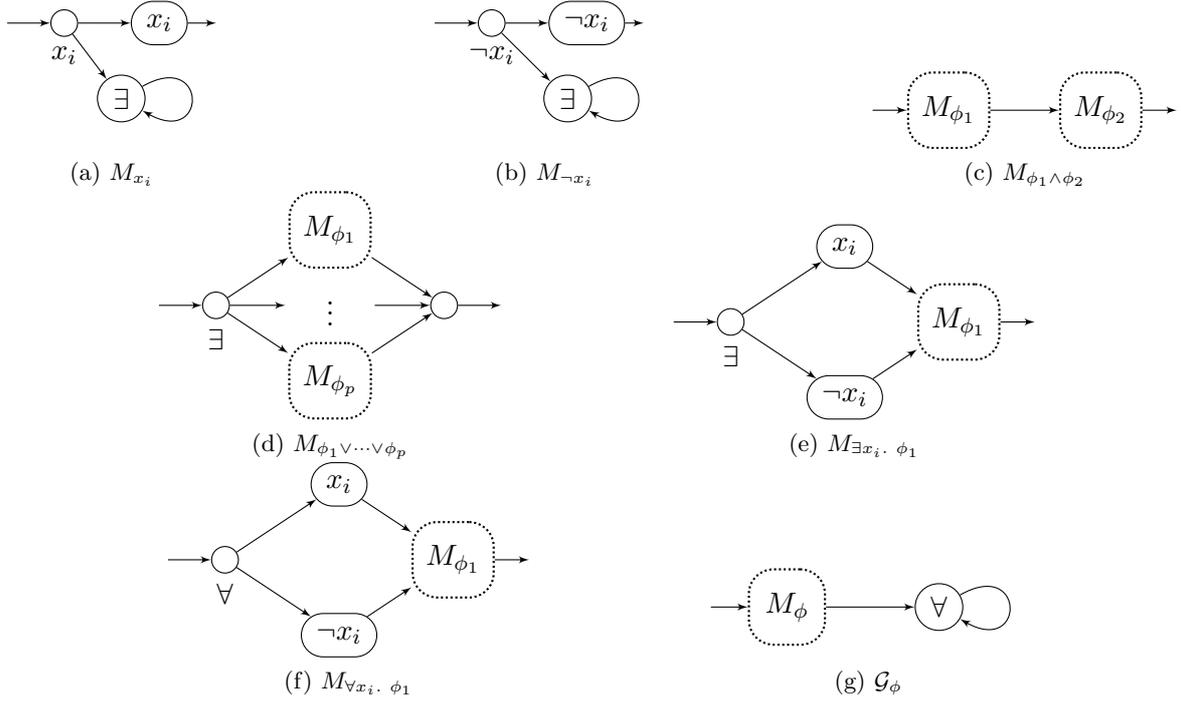
\begin{figure}
\centering
\tikzstyle{every state}+=[rounded rectangle,inner sep=5pt,minimum size=8pt]
\tikzstyle{module}+=[state,rectangle,draw,thick,densely dotted,minimum size=1cm,rounded corners=3.25mm]
\tikzstyle{every picture}+=[>=latex']
\subfloat[$M_{x_i}$]{\label{fig:Mx}
  \begin{tikzpicture}
  \draw (2,2) node[state,label=below:$x_i$] (L11) {};
  \draw (3.25,2) node[state] (M11) {$x_i$};

  \draw (2.75,1) node[state] (E) {\shortEve};

  \draw[->] (1.25,2) -- (L11);
  \draw[->] (L11) edge (M11);
  \draw[->] (L11) edge (E);
  \draw[->] (M11) edge (4,2);
  \draw[->] (E) edge[bigloop right] (E);
  \end{tikzpicture}
}
\hfill
\subfloat[$M_{\lnot x_i}$]{\label{fig:Mnx}
  \begin{tikzpicture}
  \draw (2,2) node[state,label=below:$\lnot x_i$] (L11) {};
  \draw (3.25,2) node[state] (M11) {$\lnot x_i$};

  \draw (3,1) node[state] (E) {\shortEve};

  \draw[->] (1.25,2) -- (L11);
  \draw[->] (L11) edge (M11);
  \draw[->] (L11) edge (E);
  \draw[->] (M11) edge (4,2);
  \draw[->] (E) edge[bigloop right] (E);
  \end{tikzpicture}
}
\hfill
\subfloat[$M_{\phi_1 \land \phi_2}$]{\label{fig:Mand}
  \begin{tikzpicture}
  \draw (1,2) node[module] (M1) {$M_{\phi_1}$};
  \draw (3,2) node[module] (M2) {$M_{\phi_2}$};

  \draw[->] (0,2) -- (M1);
  \draw[->] (M1) edge (M2);
  \draw[->] (M2) edge (4,2);
  \end{tikzpicture}
}

\hfill~
\subfloat[$M_{\phi_1 \lor \dots \lor \phi_p}$]{\label{fig:Mor}
  \begin{tikzpicture}
  \draw (0,1) node[state,label=below:\shortEve] (I) {};
  \draw (1.5,2) node[module] (M1) {$M_{\phi_1}$};
  \draw (1.5,1) node[inner xsep=15pt] (M2) {$\vdots$};
  \draw (1.5,0) node[module] (MP) {$M_{\phi_p}$};
  \draw (3,1) node[state] (F) {};
  
  \draw[->] (-0.75,1) -- (I);
  \draw[->] (I) edge (M1);
  \draw[->] (I) edge (M2);
  \draw[->] (I) edge (MP);
  \draw[->] (M1) edge (F);
  \draw[->] (M2) edge (F);
  \draw[->] (MP) edge (F);
  \draw[->] (F) edge (3.75,1);
  \end{tikzpicture}
}
\hfill~
\subfloat[$M_{\exists x_i.\ \phi_1}$]{\label{fig:Mexists}
  \begin{tikzpicture}
  \draw (0,1) node[state,label=below:\shortEve] (I) {};
  \draw (1.5,2) node[state] (X) {$x_i$};
  \draw (1.5,0) node[state] (NX) {$\lnot x_i$};
  \draw (3,1) node[module] (M1) {$M_{\phi_1}$};
  
  \draw[->] (-0.75,1) -- (I);
  \draw[->] (I) edge (X);
  \draw[->] (I) edge (NX);
  \draw[->] (X) edge (M1);
  \draw[->] (NX) edge (M1);
  \draw[->] (M1) edge (4,1);

  \end{tikzpicture}
}
\hfill~

\hfill~
\subfloat[$M_{\forall x_i.\ \phi_1}$]{\label{fig:Mforall}
  \begin{tikzpicture}
  \draw (0,1) node[state,label=below:\shortAdam] (I) {};
  \draw (1.5,2) node[state] (X) {$x_i$};
  \draw (1.5,0) node[state] (NX) {$\lnot x_i$}; 
  \draw (3,1) node[module] (M1) {$M_{\phi_1}$};

  \draw[->] (-0.75,1) -- (I);
  \draw[->] (I) edge (X);
  \draw[->] (I) edge (NX);
  \draw[->] (X) edge (M1);
  \draw[->] (NX) edge (M1);
  \draw[->] (M1) edge (4,1);
  \end{tikzpicture}
}
\hfill~
\subfloat[$\G_{\phi}$]{\label{fig:G}
  \begin{tikzpicture}
    \draw (3,1) node[module] (M1) {$M_{\phi}$};
    \draw (5,1) node[state] (F) {\shortAdam};

    \draw[->] (2,1) -- (M1);
    \draw[->] (M1) edge (F);
    \draw[->] (F) edge[bigloop right] (F);
  \end{tikzpicture}
}
\hfill~
\caption{Modules for the definition of the game $\G_\phi$.}
\end{figure}

Given a history of the game, 
we write $\lost(h)=\{ p \mid \exists i.\ h_i \in \bad_p \}$ 
the set of player who already lost on that path.
We associate to such a set of players~$\lost$, a partial
valuation~$v_\lost$ such that:
\[
v_\lost(x_i) = \left\{\begin{array}{ll}
1 & \textrm{if } x_i \in \lost \\
0 & \textrm{if } \lnot x_i \in \lost \wedge x_i \notin \lost \\
\textrm{undefined} & \textrm{otherwise} \\
\end{array}\right.
\]
Note that upon entering $M_\psi$, the module for the propositional part of $\phi$, $v_\lost$ is a valuation over variables $\{x_i\}_i$: each variable had exactly one literal visited.
We write $v_{\exists\forall}$ this valuation.

\medskip

Let $\ell$ be a literal.
Let $S_\ell^1$ be the set of strategies of player $\ell$ admissible with respect to all the other strategies, \emph{i.e.} $S_\ell^1$ is the set of admissible strategies after one elimination phase.
We claim that $S_\ell^1$ is exactly the set of strategies that, given a history $h$ such that $\ell \notin \lost(h)$ and $\last(h) \in V_\ell$, do not take the transition to state labeled $\ell$.

Indeed, if such a transition is taken by strategy $\sigma_i$, player $x_i$ loses while the strategy $\sigma_i'$ that mimics $\sigma_i$ up to this $M_{x_i}$ module then chooses the other transition is winning in this case and preforming in the same way otherwise, thus $\sigma_i' \domstr{\stratset} \sigma_i$.

Now if $\sigma_i$ is a strategy that never chooses the transition to states labeled $\ell$ unless $\ell \in \lost(h)$.
Assume $\sigma_i' \domstr{\stratset} \sigma_i$.
Let $\sigma_{-i}$ be a profile for the other players such that $\win{\ell}(\sigma_i',\sigma_{-i})$ and $\neg \win{\ell}(\sigma_i,\sigma_{-i})$.
Consider the first time $\sigma_i$ and $\sigma_i'$ diverged in these plays, after a history $h$.
If $\ell \in \lost(h)$, both plays are losing, hence it is the case that $\ell \notin \lost(h)$.
According to its constraints, $\sigma_i$ chooses to go state losing for \Eve, and thus wins since no other state, hence no state labeled $\ell$, is ever visited, which is a contradiction.

Remark that now all the choices that remain for player $\ell$ is when he has already lost.
Thus the set of admissible strategies for player $\ell$ is stabilized for these players after one iteration.

The set of strategies $S_\Eve^1$ and $S_\Adam^1$ of admissible strategies for the first iterations for players \Eve and \Adam, respectively, are identical to the initial sets of strategies for these players.
Indeed, \Eve can be made to lose or win by the coalition of other players regardless of her choices.
Since the objectives of \Adam and \Eve are opposite, this is also the case for \Adam.
For any history that is not already in a sink state where \Eve has lost (then nobody has choices), other players can play the two following profiles:
\begin{enumerate}
\item To make \Eve win:
\begin{itemize}
\item when in a module $M_\ell$, player $\ell$ always chooses to go to the state labeled $\ell$ (even though he loses);
\item what \Adam plays is irrelevant.
\end{itemize}
This ensures progress toward the end state where \Eve wins (and \Adam loses).
\item To make \Eve lose:
\begin{itemize}
\item when in a module $M_\ell$, player $\ell$ always chooses to go to the state labeled $\exists$;
\item what \Adam plays is irrelevant.
\end{itemize}
Since these modules are encountered for every path \Eve may choose, she loses\footnote{Unless $\phi$ is the empty formula, in which case the game is only the sink state, thus there are no real strategies.} (and \Adam wins).
\end{enumerate}

\subparagraph{First, assume that $\phi$ is satisfiable.}
We then show that the admissible strategies of \Eve in the second phase of strategy elimination, $S_\Eve^2$, are the ones corresponding to a satisfaction of $\phi$.
Namely, the choices of \Eve in choosing the value of existentially quantified variables should yield a true propositional formula $\psi$, which truth is proved by \Eve choosing the proper clause in the case of disjunctions.
Indeed, if \Eve follows such a strategy the only modules $M_\ell$ entered are ones where literal $\ell$ has been set as true by valuation $v_{\exists\forall}$.
So player $\ell$ has already lost in the history of the play, hence transitions to states labeled by $\ell$ are possible.
Since it is still possible to make \Eve lose from these states, one can still devise strategy profiles of other players that cooperate only with a given strategy of \Eve, thus ensuring that it is not dominated.

On the other hand, if a false literal $\ell$ is encountered, which will occur if the valuation $v_{\exists\forall}$ does not satisfy $\psi$ or if \Eve tries to validate a false part of a disjunction, then the player $\ell$ has no choice but to make \Eve lose, since he must play a strategy of $S_\ell^1$.
So all strategies that do not correspond to the validation of $\phi$ are losing even with cooperation of other players, so they are not admissible\footnote{Since some admissible strategy exists.}.

Similarly, admissible strategies of \Adam, $S_\Adam^2$, require him to choose, if possible (that means if \Eve picked the wrong value for a variable), values for universally quantified variables that makes $\psi$ false.\MS{Dans ce cas c'est dr\^ole mais les admissibles ressemblent \`a des subgame perfects. C'est exactement elles? Y a-t-il des liens?}

The iteration of admissible strategies stops after these two steps.
Indeed, the only difference from \Eve's point of view is the strategies of \Adam, which cannot change the result provided \Eve plays to satisfy $\phi$.
Conversely, \Adam has no better choice than to try to falsify the formula, even if this case does not arise anymore since strategies of $S_\Eve^2$ do not make such mistakes.

Therefore the iteratively admissible strategies of $\G_\phi$ are the ones in $S_\Eve^2 \times S_\Adam^2 \times S_{x_1}^1 \times S_{\lnot x_1}^1\times\cdots$.
And for any strategy of $S_\Eve^2$, the strategy profile such that all literal players chose to go to $\ell$ whenever possible (\emph{i.e.} whenever they have already lost) is winning for \Eve.

\subparagraph{Now assume that $\phi$ is not satisfiable.}
%
On the other hand, strategies for \Adam are not dominated, $S_\Adam^2$, are actually winning: \Adam can choose to play a valuation such that $v_{\exists\forall} \nvDash \psi$.
With such a valuation, for any choice of \Eve in the disjunctions, a module $M_\ell$ is eventually encountered where $v_{\exists\forall}(\ell)=0$, so player $\ell$ has no choice but to make \Eve lose, or equivalently to make \Adam win.

Any further elimination is thus useless: any iteratively admissible profile has a strategy for \Adam that makes him win, so \Eve loses.
\qedhere
\end{proof}
\end{collect}

\begin{figure*}[htbp]
\centering
\begin{tikzpicture}[yscale=1,xscale=1.2]
\tikzstyle{every state}+=[rounded rectangle,inner sep=5pt,minimum size=8pt]
\tikzstyle{module}+=[state,rectangle,draw,thick,densely dotted,minimum size=1cm,rounded corners=3.25mm]
\tikzstyle{every picture}+=[>=latex']

\begin{scope}[rotate=90]

\node[state,label=below:\shortEve] (set1) at (0,-2.1) {};
\node[state] (t1) at (1.2,-2.75) {$x_1$};
\node[state] (f1) at (-1.2,-2.75) {$\lnot x_1$};

\node[state,label=below:\shortAdam] (set2) at (0,-3.4) {};
\node[state] (t2) at (1.2,-4.05) {$x_2$};
\node[state] (f2) at (-1.2,-4.05) {$\lnot x_2$};

\node[state,label=below:\shortEve] (set3) at (0,-4.7) {};
\node[state] (t3) at (1.2,-5.35) {$x_3$};
\node[state] (f3) at (-1.2,-5.35) {$\lnot x_3$};

\node[state,label=below:\shortEve] (ch1) at (0,-6) {};
\node[state,label=below:$x_1$] (ch1a1) at (2,-7) {};
\node[state,label=below:$x_2$] (ch1a2) at (0,-7) {};
\node[state,label=below:$\lnot x_3$] (ch1a3) at (-2,-7) {};

\node[state] (l11) at (2.75,-8.25) {$x_1$};
\node[state] (l1e1) at (1.75,-8.25) {$\shortEve$};
\node[state] (l12) at (0.75,-8.25) {$x_2$};
\node[state] (l1e2) at (-0.75,-8.25) {$\shortEve$};
\node[state] (l13) at (-1.75,-8.25) {$\lnot x_3$};
\node[state] (l1e3) at (-2.75,-8.25) {$\shortEve$};

\node[state,label=below left:\shortEve] (ch2) at (0,-9.5) {};
\node[state,label=below:$\ \ \lnot x_1$] (ch2a1) at (2,-10.5) {};
\node[state,label=below:$x_2$] (ch2a2) at (0,-10.5) {};
\node[state,label=below:$x_3$] (ch2a3) at (-2,-10.5) {};

\node[state] (l21) at (2.75,-11.75) {$\lnot x_1$};
\node[state] (l2e1) at (1.75,-11.75) {$\shortEve$};
\node[state] (l22) at (0.75,-11.75) {$x_2$};
\node[state] (l2e2) at (-0.75,-11.75) {$\shortEve$};
\node[state] (l23) at (-1.75,-11.75) {$x_3$};
\node[state] (l2e3) at (-2.75,-11.75) {$\shortEve$};

\node[state] (ad) at (0,-13) {$\shortAdam$};

\draw[-latex'] (0,-1.6) -- (set1);
\draw[-latex'] (set1) edge (t1);
\draw[-latex'] (t1) edge (set2);
\draw[-latex'] (set1) edge (f1);
\draw[-latex'] (f1) edge (set2);

\draw[-latex'] (set2) edge (t2);
\draw[-latex'] (t2) edge (set3);
\draw[-latex'] (set2) edge (f2);
\draw[-latex'] (f2) edge (set3);

\draw[-latex'] (set3) edge (t3);
\draw[-latex'] (t3) edge (ch1);
\draw[-latex'] (set3) edge (f3);
\draw[-latex'] (f3) edge (ch1);

\draw[-latex'] (ch1) edge (ch1a1);
\draw[-latex'] (ch1) edge (ch1a2);
\draw[-latex'] (ch1) edge (ch1a3);

\end{scope}
\begin{scope}
\tikzstyle{loop below}=[loop right]

\draw[-latex'] (ch1a1) edge (l11);
\draw[-latex'] (ch1a1) edge (l1e1);
\draw[-latex'] (l1e1) edge[loop below] (l1e1);

\draw[-latex'] (ch1a2) edge (l12);
\draw[-latex'] (ch1a2) edge (l1e2);
\draw[-latex'] (l1e2) edge[loop below] (l1e2);

\draw[-latex'] (ch1a3) edge (l13);
\draw[-latex'] (ch1a3) edge (l1e3);
\draw[-latex'] (l1e3) edge[loop below] (l1e3);

\draw[-latex'] (l11) edge[bend left=20] (ch2);
\draw[-latex'] (l12) edge (ch2);
\draw[-latex'] (l13) edge[bend right=20] (ch2);

\draw[-latex'] (ch2) edge (ch2a1);
\draw[-latex'] (ch2) edge (ch2a2);
\draw[-latex'] (ch2) edge (ch2a3);

\draw[-latex'] (ch2a1) edge (l21);
\draw[-latex'] (ch2a1) edge (l2e1);
\draw[-latex'] (l2e1) edge[loop below] (l2e1);

\draw[-latex'] (ch2a2) edge (l22);
\draw[-latex'] (ch2a2) edge (l2e2);
\draw[-latex'] (l2e2) edge[loop below] (l2e2);

\draw[-latex'] (ch2a3) edge (l23);
\draw[-latex'] (ch2a3) edge (l2e3);
\draw[-latex'] (l2e3) edge[loop below] (l2e3);

\draw[-latex'] (l21) edge[bend left=20] (ad);
\draw[-latex'] (l22) edge (ad);
\draw[-latex'] (l23) edge[bend right=20] (ad);
\draw[-latex'] (ad) edge[loop below] (ad);

\end{scope}
\end{tikzpicture}
\caption{Game $\G_{\mu}$ with $\mu = \exists x_1 \forall x_2 \exists x_3 (x_1 \lor x_2 \lor \lnot x_3) \land (\lnot x_1 \lor x_2 \lor x_3)$. A label $y$ \emph{inside} a state $s$ denotes that $s\in\bad_y$; a label $y$ \emph{below} a state $s$ denotes that $s\in V_y$. Note that \eve is abbreviated to \shortEve\ and \adam is abbreviated to \shortAdam.}
\label{fig:exsafetyhardness}
\end{figure*}
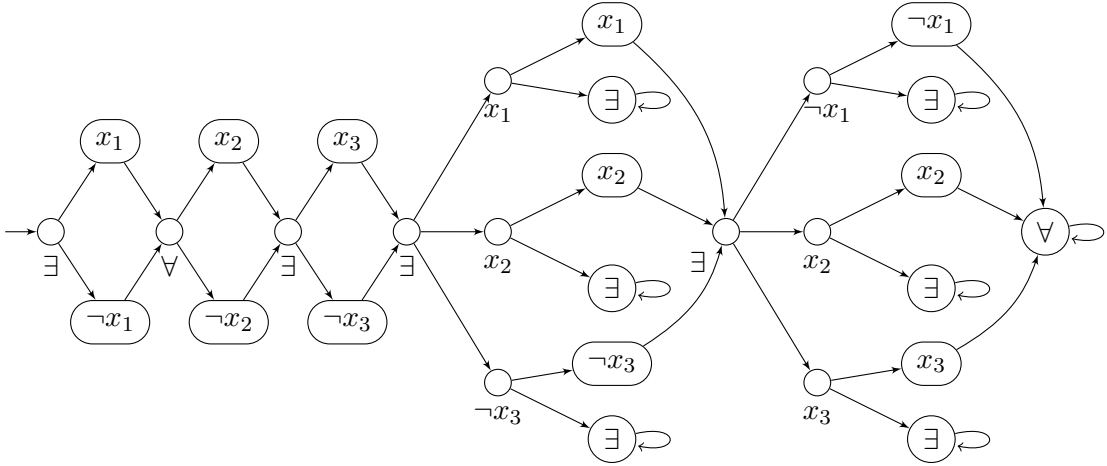

\begin{collect}{appendix-safety}{}{}
As shown in the above proof, the complexity stems from the number of players rather that from the number of necessary iteration needed to reach the set of iteratively admissible strategies:
\begin{corollary}\MS{C'est en fait un scholie :-)}
Deciding whether there exists a profile $\profile$ of strategies admissible after $2$ iterations such that $\states{\outcome(\profile)} \cap \bad_i = \emptyset$ is \PSPACE-hard.
\end{corollary}

\end{collect}

%% file: transformation.tex
Our main result for this section is stated in the following theorem:
\begin{theorem}\label{thm:circuitPspace}
The winning coalition problem with a Muller condition for each player is \PSPACE-complete.
The problem is \PSPACE-hard even when restricted to two players.
\end{theorem}

\PSPACE-hardness follows from \PSPACE-hardness of two-player games with Muller conditions~\cite{dawar13}.

The idea of the algorithm 
is to construct 
a graph representation of the outcomes of admissible profiles.
While the construction also relies on the notion of value, it is more involved than for safety conditions.



\ifShort\else
After proving Theorem~\ref{thm:circuitPspace}, we investigate the special case of B\"uchi winning conditions, for which we obtain a better complexity:
\begin{theorem}\label{thm:Buchi}
The winning coalition problem with B\"uchi objectives is in $\NP\cap\coNP$. 
Moreover, if there exists a polynomial algorithm for solving two-player parity games, then winning coalition problem with B\"uchi objectives is in~$\P$.
\end{theorem}

We also use the fact that this construction provides automata representing the outcomes of all iteratively admissible profiles to solve the model-checking under admissibility problem with tight bounds:
\begin{theorem}\label{thm:LTLmcAdm}
The model-checking under admissibility problem is \PSPACE-complete.
\end{theorem}
\fi

\begin{collect}{appendix-prefix-independent}{}{}

\subsection{Shifting of strategy profiles}

In the case of prefix-independent objectives, each player can, at any time change his mind.
In fact, it is possible for a strategy to \emph{shift} to another (admissible) strategy.

\begin{definition}[Shifting]
  Given two strategies $\sigma^1$ and $\sigma^2$, and a history~$h$,
  we denote by $\sigma^1\switch{h}{\sigma^2}$ the strategy~$\sigma$
  that follows strategy~$\sigma^1$ and \newdef{shifts} to~$\sigma^2$
  after history~$h$. 
  Formally, given a history $h'$:
\[\sigma(h') = \left\{\begin{array}{ll}
\sigma^2(h^{-1}h') & \textrm{if } h \prefix h' \\
\sigma^1(h') & \textrm{otherwise}
\end{array}\right.\]
  We say that a strategy set~$S$ \newdef{allows shifting}, if for any 
  $\sigma^1$ and $\sigma^2$ in~$S$, and every history $h$ we have
  $\sigma^1\switch{h}{\sigma^2} \in S$.
  A rectangular set of profiles allows shifting if all its components
  do. 
\end{definition}

\begin{definition} 
  We write $\sigma_i \circ h$ the strategy:
  \[\sigma_i\circ h (h') = \left\{\begin{array}{ll}
  \sigma_i(h \cdot h'_{\ge 1}) & \text{if}\ h'_0 = \last(h)\\
  \sigma_i(h') & \text{otherwise}
  \end{array}\right. \]
\end{definition}

Strategy $\sigma_i \circ h$ thus plays as if $\sigma_i$ was played after a history $h$.
However, if the history does not start at the end state of $h$, this would be inconsistent; in that case $\sigma_i \circ h$ plays like $\sigma_i$ 

\begin{lemma}\label{lem:circ}
  If $\stratset^n$ allows shifting, $h$ is compatible with a strategy profile
  $(\sigma_j)_{j\in P}$ of $\stratset^n$ and $\sigma_i$ is in $\stratset^{n+1}_i$ then
  $\sigma_i \circ h$ is in $\stratset^{n+1}_i$.
\end{lemma}
\begin{proof}
  Assuming $\sigma_i$ admissible with respect to $\stratset^n$, we prove that $\sigma_i \circ h$ is also admissible, \ie not weakly dominated.
  

  Consider a strategy $\sigma'_i\in \stratset^n_i$.
  We show that $\sigma'_i \not\domstr{\stratset^n} \sigma_i \circ h$, which shows
    that $\sigma_i \circ h$ is in $\stratset^{n+1}_i$.
  First, if $\sigma_i \circ h \dom{\stratset^n} \sigma_i'$, then clearly $\sigma_i' \not\domstr{\stratset^n} \sigma_i \circ h$.

  Otherwise $\sigma_i \circ h \not\dom{\stratset^n} \sigma_i'$.
  We write $\sigma''_{i} = \sigma_i\switch{h}{\sigma'_i}$ for the strategy
  profile that plays according to $\sigma_{i}$ until history $h$ and then shifts
  to $\sigma'_{i}$. 

  Since $\sigma'_i$ is not weakly dominated by $\sigma_i \circ h$, there is a state $s$ and a profile
  $\tau_{-i}\in \stratset^n_{-i}$ such that $\outcome_s(\tau_{-i},\sigma'_i) \in \win{i}$ and $\outcome_s(\tau_{-i},\sigma_i\circ h)$ does not win for $i$.
  Note that since $\sigma_i$ and $\sigma_i \circ h$ are identical except when starting from state $\last(h)$, it must be the case that $s=\last(h)$.
  
  We define the strategy profile $\tau'_{-i}=\sigma_{-i}\switch{h}{\tau_{-i}}$ that
  follows $\sigma_{-i}$ and until $h$ then shifts to $\tau_{-i}$. \MS{Formellement ce serait plut\^ot: ``$\tau_{-i}'=(\tau_j')_{j \in \Agt\setminus\{i\}}$ defined for all $j \neq i$ by $\tau'_{j}=\sigma_{j}\switch{h}{\tau_{j}}$'', mais on comprend quand m\^eme.}
  Recall that $\sigma_{-i}$ is a profile of strategies of $\stratset^n_{-i}$ such that $h \prefix \outcome_{h_0}(\sigma_{-i},\sigma_i)$.
  We have that $\outcome_{h_0}(\tau'_{-i},\sigma_i) =
  h\cdot\outcome_{\last(h)}(\tau_{-i},\sigma_i\circ h)$ is losing for \player{i}, and
  $\outcome_{h_0}(\tau'_{-i},\sigma''_i) = h\cdot\outcome_{\last(h)}(\tau_{-i},\sigma'_i)$
  is winning for \player{i}, by prefix independence. 
  Moreover $\tau'_{-i}\in \stratset^n_{-i}$ because $\stratset^n$ allows shifting by hypothesis.
  So $\sigma_i$ does not weakly dominate $\sigma''_i$: $\sigma_i \not\dom{\stratset^n} \sigma''_i$.

  Since $\sigma_i$ is not weakly dominated, in particular $\sigma_i'' \not\domstr{\stratset^n} \sigma_i$.
  And by the above result $\sigma_i \not\dom{\stratset^n} \sigma''_i$, so it must be the case that
  $\sigma''_i \not\dom{\stratset^n} \sigma_i$.
  Hence, there is a strategy profile~$\tau_{-i}$ and a state $s'$ such that
  $\outcome_{s'}(\tau_{-i},\sigma_i)$ is winning for \player{i} and
  $\outcome_{s'}(\tau_{-i},\sigma''_i)$ is losing for \player{i}. 
  Since $\sigma_i$ and $\sigma''_i$ are identical except after history $h$, this
  means that $s'=h_0$ and that $h$ is a prefix of both $\outcome_{h_0}(\tau_{-i},\sigma_i)$ and
  $\outcome_{h_0}(\tau_{-i},\sigma''_i)$.
  We therefore have that $\outcome_{\last(h)}(\tau_{-i} \circ h,\sigma_i \circ h)$ is
  winning for \player{i} but $\outcome_{\last(h)}(\tau_{-i} \circ h,\sigma''_i \circ h)$ is
  losing for $i$ from $\last(h)$.
  Since $\sigma''_i \circ h = \sigma'_i$, this means that $\sigma'_i$ does not
  dominate $\sigma_i \circ h$.
\qedhere
\end{proof}

\begin{lemma} \label{lem:shifting}
  For any integer $n$, $\stratset^n$ allows shifting.
\end{lemma}
\begin{proof}
  The proof is by induction over $n$.
  For the case $n=0$, the property obviously holds since
  all strategies are in $\stratset^0$.
  Assuming the property holds for $n$, we show it holds for $n+1$.

  Let $\sigma^1$ and $\sigma^2$ be two strategies of $\stratset^{n+1}_i$ and $h$ a history compatible with a profile of $\stratset^n$
  .
  First, if there exists a winning strategy against all profiles of $\stratset_{-i}^n$, then all strategies of $\stratset^{n+1}_i$ are winning strategies.
  In particular it is the case for both $\sigma^1$ and $\sigma^2$.
  By prefix-independence, it is also the case for $\sigma^1\switch{h}{\sigma^2}$, hence $\sigma^1\switch{h}{\sigma^2} \in \stratset^{n+1}_i$.
  Therefore we assume that there is no winning strategy against all profiles of $\stratset_{-i}^n$.
  
  Let $\sigma^3 \in \stratset^n_i$, we show that $\sigma^3 \not\domstr{\stratset^n} \sigma^1\switch{h}{\sigma^2}$.
  Without loss of generality, $\sigma^3$ can be assumed admissible: $\sigma^3 \in \stratset^{n+1}_i$.
Lemma~\ref{lem:circ} ensures that $\sigma^2 \circ h \in \stratset^{n+1}_i$ and $\sigma^3 \circ h \in \stratset^{n+1}_i$.
Since they are both admissible though not winning strategies, there exists a profile $\tau_{-i}$ such that $\sigma^2 \circ h$ wins while $\sigma^3 \circ h$ loses from $\last(h)$.
Since $\sigma^1$ and $\sigma^3$ are also both admissible, there is also a profile $\tau_{-i}'$ that makes $\sigma^1$ win and $\sigma^3$ lose.
Consider now profile $\tau_{-i}''=\tau_{-i}'\switch{h}{\tau_{-i}}$.
We show that $\rho = \outcome_{h_0}(\tau_{-i}'',\sigma^3)$ is losing for \player{i} while $\rho'= \outcome_{h_0}(\tau_{-i}'',\sigma^1\switch{h}{\sigma^2})$ is winning for \player{i}.

If $h$ is not a prefix of the run $\outcome_{h_0}(\tau_{-i}',\sigma^3)$, then this losing outcome is exactly $\rho$.
If $h$ is a prefix of $\outcome_{h_0}(\tau_{-i}',\sigma^3)$, then $\rho = h \cdot \outcome(\tau_{-i},\sigma^3 \circ h)$ which is a losing outcome, by prefix independence.

Similarly, if $h$ is not a prefix of $\outcome_{h_0}(\tau_{-i}',\sigma^1)$, then this winning outcome is exactly $\rho'$.
And if $h$ is a prefix of $\outcome_{h_0}(\tau_{-i}',\sigma^1)$, then $\rho'= h \cdot \outcome(\tau_{-i},\sigma^2 \circ h)$ which is a winning outcome for \player{i}, by prefix independence.
\qedhere
\end{proof}
\end{collect}

\begin{collect*}{appendix-prefix-independent}{}{\label{sec:char}}{}{}
\subsection{Characterizing outcomes of admissible strategies using the sequence of their value}
\end{collect*}

\begin{proposition}\label{prop:valuestate}
For prefix-independent objectives, the value depends only on the last state of the history.
\end{proposition}
\begin{collect}{appendix-prefix-independent}{\subsubsection{Proof of Proposition~\ref{prop:valuestate}}}{}
\begin{proof}
Let $h,h'$ be two different histories with $\last(h)=\last(h')$.
Assume $\val^{n+1}_i(h)=1$.
Then $i$ has a strategy $\sigma_i$ starting from that wins against all profiles of $\stratset^n_{-i}$.
Then playing like $\sigma_i$ after $h'$, namely switching to $\sigma_i \circ h$ which is a strategy of $\stratset^n_i$ by Lemma~\ref{lem:circ}, is still a winning strategy so $\val^{n+1}_i(h)=1$.

Assume $\val^{n+1}_i(h)=-1$ and by contradiction that $\val^{n+1}_i(h')>-1$.
Then there exists a profile $\profile \in \stratset^n$ such that $\profile$ wins from $h'$.
So playing like $\profile$ after $h$, namely playing $\profile \circ h' \in \stratset^n$, is also winning w.r.t. $\win{i}$.
Therefore it is not the case that $\val^{n+1}_i(h)=-1$.

The case of $\val^{n+1}_i(h)=0$ is then obtained by definition of value.
\end{proof}
\end{collect}
Hence we write $\val^{n}_j(s)$ instead of $\val^{n}_j(h)$ when \mbox{$\last(h)=s$}.
Since $\val^n_j$ is here a function from $V$ to $\{-1,0,1\}$, it can be extended to runs:
 \(\hval{n}{j}(\rho)\) is the word $w \in \{-1,0,1\}^\omega$ such that $w_k = \val^n_j(\rho_k)$ for all $k$.


\begin{collect}{appendix-prefix-independent}{}{}
\begin{lemma}\label{lem:valeqadm}
For all $n$, if $s \in V_i$, $\trans{s}{s'}$ with $\val^n_i(s) = \val^n_i(s')$, then for any strategy $\sigma_i \in \stratset^{n+1}$ that is admissible, the strategy $\sigma_i'$ defined by
\begin{mathpar}
\sigma_i'(s) = s'
\and
\sigma_i'(s \cdot h) = \sigma_i(h)
\and
\sigma_i'(h) = \sigma_i(h) \textrm{ if } h_0 \neq s
\end{mathpar}
is also an admissible strategy of $\stratset^{n+1}$.
\end{lemma}
\begin{proof}
The case where $\val^n_i(s)=1$ (resp. $\val^n_i(s)=-1$) is trivial since $\sigma_i$ is a winning (resp. always losing) strategy, then so is $\sigma_i'$.

Otherwise, assume $\val^n_i(s) = \val^n_i(s') = 0$.
Let $\tau_i \in \stratset^n_i$.
We only consider what happens in $s$, since otherwise $\sigma_i'$ behaves the same as $\sigma_i$.
If $\tau_i(s) = s'' \neq s'$, then $\val^n_i(s'') \leq 0$ so there is a profile that makes $\tau_i$ lose from $s''$ hence from $s$.
Additionally, a profile exists that makes $\sigma_i$ win from $s'$, hence it makes $\sigma_i'$ win from $s$.

If $\tau_i(s) = s'$.
Since $\sigma_i$ is admissible, $\tau_i \circ s \not\domstr{\stratset^n} \sigma_i$.
\begin{itemize}
\item If there is a profile that makes $\sigma_i$ win but $\tau_i \circ s$ lose, it also makes $\sigma_i'$ win and $\tau_i$ lose.
\item Otherwise every profile that makes $\tau_i \circ s$ win also makes $\sigma_i$ win, hence all profiles that make $\tau_i$ win also make $\sigma_i'$ win.\qedhere
\end{itemize}
\end{proof}

\begin{lemma}\label{lem:value}
  For all integer $n$ and states $s,s'$ such that $\trans{s}{s'}$,
  if $s$ is controlled by $i$ then $\val^n_i(s) \ge \val^n_i(s')$.
\end{lemma}
\begin{proof}
  Assume $s$ is controlled by $i$:
  \begin{itemize}
  \item if $\val^n_i(s')=-1$ then the property is obviously true.
  \item if $\val^n_i(s')=0$ then there exists a strategy profile $\sigma_\Agt
    \in \stratset^n$ that makes \player{i} win from $s'$.
    We define the strategy $\sigma'_i$ that plays from $s$ to $s'$ and then
    follows $\sigma_i$: 
    \begin{mathpar}
    \sigma_i'(s) = s'
    \and
    \sigma_i'(s \cdot h) = \sigma_i(h)
    \and
    \sigma_i'(h) = \sigma_i(h) \textrm{ if } h_0 \neq s
    \end{mathpar}
    Strategy $\sigma'_i$ is a strategy of $\stratset^n_i$ by Lemma~\ref{lem:valeqadm}.
    The strategy profile $\sigma_{j\ne i}\switch{s \cdot s'}{\sigma_{j\ne i}}$ is in $\stratset^n$ because of Lemma~\ref{lem:shifting}.
    The outcome of $(\sigma'_i,\sigma_{j\ne i}\switch{s \cdot s'}{\sigma_{j\ne
        i}})$ is then winning from $s$ by prefix independence.
    The value of $s$ is therefore at least $0$.
  \item if $\val^n_i(s')=1$ then there 
    exists a strategy $\sigma_i\in \stratset^n_i$ that is winning for \player{i} from
    $s'$ for all the strategy profiles in $\prod_{i\ne j} \stratset^n_j$.
    Consider the strategy $\sigma'_i$ that plays from $s$ to $s'$ and then
    follows $\sigma_i$.
    Let $\sigma'_{j\ne i}$ be a strategy profile in $\prod_{i\ne j} \stratset^n_j$.
    By Lemma~\ref{lem:circ}, $\sigma'_{j\ne i} \circ (s \cdot s')$ is a strategy
    profile in $\prod_{i\ne j} \stratset^n_j$. \JFR{If not, this is not a problem, right?}\MS{Hum... il faut quand m\^eme qu'on s'assure que $\sigma_i$ gagne depuis $s'$.}
    Strategy~$\sigma_i$ is winning against  $\sigma'_{j\ne i} \circ (s \cdot s')$ from $s'$ and by
    prefix independence $\sigma'_i$ is winning against $\sigma'_{j\ne i}$ from $s$.
    Therefore the value of $s$ is $1$.\qedhere
  \end{itemize}
\end{proof}

\begin{lemma}\label{lem:admvaleq}
  For all $n$, if $s \in V_i$ and $\sigma_i \in \stratset^{n+1}_i$ then for any history $h$, $\val^n_i(s) = \val^n_i(\sigma_i(h\cdot s))$.
\end{lemma}
\begin{proof}
By the previous lemma the value cannot increase.
If \player{i} plays an admissible strategy $\sigma_i \in \stratset^{n+1}_i$, we show that it cannot decrease.
Let $s' = \sigma_i(h\cdot s)$
\begin{itemize}
\item If $\val^n_i(s)=1$, then $\sigma_i$ is a winning strategy.
Since there is no such strategy from a state with value $\val^n_i \leq 0$, $\val^n_i(s')=1$.
\item If $\val^n_i(s)=0$, then there is a profile $\sigma_{-i} \in \stratset^n_{-i}$ such that $\rho = \outcome(\sigma_i,\sigma_{-i}) \in \win{i}$.
Note that $h \cdot s \cdot s' \prefix \rho$.
If $\val^n_i(s')=-1$, there can be no such profile, thus $\val^n_i(s')>-1$, hence $\val^n_i(s')=0$
\item If $\val^n_i(s)=-1$, there can be no lower value so $\val^n_i(s')=-1$.\qedhere
\end{itemize}
\end{proof}
\end{collect}

The following lemma shows that in terms of the sequence of value for a player, we can distinguish three types of plays that are outcome of admissible strategies. 
\begin{lemma}\label{lem:disjonctionval}
  Let $s\in V$, $\rho\in \outcome_s(\stratset^n)$, and $i\in\Agt$. 
  \[ \text{If }\rho \in \outcome_s(\stratset^{n+1}_i) \text{ then }
  \hval{n}{i}(\rho) \in 0^* 1^\omega + 0^\omega + 0^*(-1)^\omega .\]
\end{lemma}
\begin{collect}{appendix-prefix-independent}{\subsubsection{Proof of Lemma~\ref{lem:disjonctionval}}}{}
\begin{proof}
  Let $\rho$ be the outcome of a strategy profile $\profile\in \stratset^n$ and $k\geq 0$ be an index.
  \begin{itemize}
  \item If $\val^n_i(\rho_k) = -1$ and $\val^n_i(\rho_{k+1}) \ge 0$.
    There is a strategy profile~$\profile'$ in $\stratset^n$ that is winning for \player{i} from $\rho_{k+1}$.
    By Lemma~\ref{lem:shifting} and Lemma~\ref{lem:circ}, $\profile\circ(\rho_{\le k}) \switch{\rho_{k+1}}{\profile'}$ is in $\stratset^n$.
    It makes \player{i} win from $\rho_k$.
    It is a contradiction with the fact that $\val^n_i(\rho_k) = -1$.
  \item If $\val^n_i(\rho_k) = 1$ and $\val^n_i(\rho_{k+1}) \le 0$.
    There is a strategy profile in $\stratset^n$ that makes $\sigma_i\circ \rho_{\le k}$ lose from $\rho_{k+1}$.
    On the other hand, since $\val^n_i(\rho_k) = 1$ there is a strategy $\sigma'_i$ that is winning for \player{i} against any strategy of $\stratset^n$.
    Hence $\sigma_i\circ \rho_{\le k}$ is dominated by $\sigma'_i$ with respect to $\stratset^n$.
    By Lemma~\ref{lem:circ} this means that $\sigma_i$ is dominated.
    Therefore $\rho$ is not compatible with a strategy in $\stratset^{n+1}_i$.
  \end{itemize}
  Therefore all the paths compatible with a profile in $\stratset^n$ and a strategy in~$\stratset^{n+1}_i$ have value of the form $0^* 1^\omega + 0^\omega + 0^*(-1)^\omega$.
  \qedhere
\end{proof}
\end{collect}

\noindent Now, we characterize outcomes of admissible strategies according to whether they end with value $1$, $-1$, or $0$. 
We do so for each player individualy.

\subparagraph{Value 1} 
To be admissible in $\stratset^n$ from a state of value $1$, a strategy has to be winning against all strategies of $\stratset^n$:
\begin{lemma}\label{lem:valun}
  Let $s\in V$, $i\in\Agt$ and $\rho \in \outcome_s(\stratset^n)$ be such that $\hval{n}{i}(\rho) \in 0^* 1^\omega$.
  \[\rho \in \outcome_s(\stratset^{n+1}_i) \text{ if, and only if, } \rho \in \win{i}.\]
\end{lemma}
\begin{collect}{appendix-prefix-independent}{\subsubsection{Proof of Lemma~\ref{lem:valun}}}{}
\begin{proof}
  We construct $\sigma_i$ that follows $\rho$ and if the history deviates from this outcome, revert to a fixed admissible strategy.
  Formally, let $\sigma^a_i \in \stratset^{n+1}_i$, we define $\sigma_i$ by:
  \begin{itemize}
  \item if $h$ is a prefix of $\rho$ with $\last(h) \in V_i$, then $\sigma_i(h) = \rho_{|h|+1}$, so that we ensure that $\rho$ is compatible with $\sigma_i$;
  \item otherwise, let $k$ be the largest index such that $h_{\le k}$ is a prefix of $\rho$, then $\sigma_i(h) = \sigma^a_i(h_{>k})$.
  \end{itemize}
  We show that $\sigma_i$ is not weakly dominated.
  Let $\sigma'_\Agt$ be a strategy profile of $\stratset^n$ whose outcome $\rho'$ is winning for \player{i} but such that $\rho'' = \outcome(\sigma'_{-i},\sigma_i)$ is not winning for \player{i}.
  We show that some profile in $\stratset^n_{-i}$ makes $\sigma'_i$ lose.
  Consider the first index $k$ such that $\rho'_k\ne \rho''_k$.
  The state $\rho'_{k-1}$ is controlled by \player{i}.

  If $\val^n_i(\rho'_{k-1}) = 0$ then there is a profile $\tau_{-i} \in \stratset^n_{-i}$ that makes $\sigma'_i$ lose from $\rho'_{k}$.
  As there is a profile $\tau_{-i}' \in \stratset^n_{-i}$ that makes $\sigma_i$ win from $\rho''_k$, we can combine this two strategy profiles to obtain one winning for $\sigma_i$ and losing for $\sigma'_i$.
  Namely, the profile $\tau^c_{-i}$ defined by
  \[\forall j \neq i, \ \tau^c_j = \sigma_j\switch{\rho_k'}{\tau_j}\switch{\rho_k''}{\tau_j'}\]
  is by construction such a profile.
  Therefore, $\sigma_i$ is not weakly dominated by $\sigma'_i$.

  Otherwise, $\val^n_i(\rho'_{k-1}) = 1$, we show that in fact $\rho''$ is a winning path which is a contradiction.
  If $\rho' = \rho$ then it is winning by hypothesis.
  Otherwise there is some point from which $\sigma_i$ plays according to $\sigma^a_i$.
  
  If this point is before index $k$.
  Since $\val^{n}_i(\rho_{k-1})=1$ there is a strategy of \player{i} winning from $\rho_{k-1}$ against all profile in $\stratset^n_{-i}$.
  As $\sigma^a_i$ is in $\stratset^{n+1}_i$, so is $\sigma^a_i \circ \rho_{\le k-1}$ by Lemma~\ref{lem:circ}.
  Therefore $\sigma^a_i \circ \rho_{\le k-1}$ should be winning from $\rho_{k-1}$ against all profiles in $\stratset^n_{-i}$.
  Hence $\rho''$ is winning for \player{i}, which is a contradiction.

  Otherwise, let $k'$ be the index at which $\sigma_i$ starts playing according to $\sigma^a_i$.
  As $\val^n_i(\rho_{k-1}) = 1$ and $k'\ge k$, $\val^n_i(\rho_{k'-1}) = 1$.
  Moreover $\rho''_{k'}$ is a successor of $\rho_{k'-1}$ which is not controlled by $i$, so $\val(\rho''_{k'}) = 1$ by Lemma~\ref{lem:value}.
  Therefore $\sigma^a_i \circ \rho_{\le k'-1}$ should be winning from $\rho''_{k-1}$ against all profile in $\stratset^n_{-i}$.
  Hence $\rho''$ is winning for \player{i}, which is a contradiction.

  This shows that $\sigma_i$ is not dominated with respect to $\stratset^n$.

  \medskip
\MS{JF: est-ce que tu veux qu'on mette les symboles $\Leftarrow$ et $\Rightarrow$?}
  Reciprocally, assume that $\sigma_i \in \stratset^{n+1}_i$.
  Let $k$ be a index such that $\val^n_i(\rho_k)=1$.
  There is a strategy of \player{i} that is winning from $\rho_k$ against all strategies of $\stratset^n$. 
  By Lemma~\ref{lem:circ}, $\sigma_i\circ \rho_{\le k}$ is also in $\stratset^{n+1}_i$.
  Therefore $\sigma_i\circ \rho_{\le k}$ also has to be winning against all strategies of $\stratset^n$. 
  As $\rho_{\ge k}$ is compatible both with a profile of $\stratset^n_{i}$ and $\sigma_i\circ \rho_{\le k}$, it is winning for \player{i}.
  Hence, by prefix independence of the objectives, $\rho$ is winning for \player{i}.
\qedhere
\end{proof}
\end{collect}

\subparagraph{Value -1}
If the run reaches a state of value $-1$, then, from there, there is no possibility of winning, so any strategy is admissible but the state of value $-1$ must not have been reached by \player{i}'s fault:
\begin{lemma}\label{lem:valmoinsun}
  Let $s\in V$, $i\in\Agt$ and $\rho \in \outcome_s(\stratset^n)$ be such that $\hval{n}{i}(\rho) \in 0^* (-1)^\omega$. Let $k$ be the index such that $\val_i^n(\rho_k) = 0 \land\val_i^n(\rho_{k+1})= -1$.
  \[ \rho \in \outcome_s(\stratset^{n+1}_i) \text{ if and only if, } \rho_k\not\in V_i  .\]
\end{lemma}
\begin{collect}{appendix-prefix-independent}{\subsubsection{Proof of Lemma~\ref{lem:valmoinsun}}}{}
\begin{proof}
  We construct a strategy $\sigma_i$ that follows $\rho$ and if the history deviates from this outcome, revert to a fixed admissible strategy, in the same way that we did in Lemma~\ref{lem:valun} for value $1$.
  We define $\sigma_i$ by:
  \begin{itemize}
  \item if $h$ is a prefix of $\rho$ with $\last(h) \in V_i$, then $\sigma_i(h) = \rho_{|h|+1}$, so that we ensure that $\rho$ is compatible with $\sigma_i$;
  \item otherwise, let $k$ be the biggest index such that $h_{\le k}$ is a prefix of $\rho$, then $\sigma_i(h) = \sigma^a_i(h_{>k})$.
  \end{itemize}
  We show that $\sigma_i$ is not weakly dominated.
  Note that starting from a state $s \neq \rho_0$ means $\sigma_i$ plays according to an admissible strategy.
  Hence only the case of plays starting from $\rho_0$ remain to be considered.
  
  Let $\sigma'_\Agt$ be a strategy profile of $\stratset^n$ whose outcome $\rho'$ is winning for \player{i} but such that $\rho'' = \outcome_{\rho_0}(\sigma'_{-i},\sigma_i)$ is not winning for \player{i}.
  Consider the first index $k$ such that $\rho'_k\ne \rho''_k$.
  The state $\rho'_{k-1} = \rho''_{k-1}$ is controlled by \player{i}.
  In order to show that $\sigma_i' \not\domstr{\stratset^n} \sigma_i$, we build a profile that makes $\sigma_i'$ lose and $\sigma_i$ win.

  If $\val^n_i(\rho'_{k-1}) = -1$, then $\rho'$ cannot be a winning path, this is a contradiction.
  Otherwise $\val^n_i(\rho'_{k-1}) = 0$.
  Then there is a profile in $\stratset^n_{-i}$ that makes $\sigma'_i$ lose from $\rho'_k$.

Now we find a profile that makes $\sigma_i$ win from $\rho''_k$.
\begin{itemize}
\item If $\rho''_{\le k-1}$ is not a prefix of $\rho$, then  $\sigma_i$ plays according to an admissible strategy.
  As $\val_i^n(\rho''_{k-1})=0$, there is a strategy profile that makes the strategy $\sigma_i\circ \rho''_{\le k}$ win.
\item Otherwise, $\rho''_{\le k-1}$ is a prefix of $\rho$.
  Since $\rho''_{k-1} = \rho_{k-1}$ is controlled by \player{i}, $\val_i^n(\rho_{k})=\val_i^n(\rho_{k-1})=0$ by hypothesis.
  Let $k'$ be the largest index such that $\val^n_i(\rho_{k'}) = 0$.
  Since $\val^n_i(\rho_{k'+1})= -1$, state $\rho_{k'}$ is not controlled by $i$.
  From $\rho_{k'}$ there is a strategy profile $\sigma''_\Agt$ that makes $\sigma^a_i$ win.
  Note that profile $\sigma_{-i}''$ deviates from $\rho$ at that point since $\val^n_i(\rho_{k'+1})=-1$: $\sigma''_{-i}(\rho_{\le k'}) \ne \rho_{k'+1}$.
  As a result, after $\rho_{\le k'}$, strategy $\sigma_i$ follows $\sigma^a_i$.
  Therefore the outcome produced by $\sigma_{-i}''$ and $\sigma_i$ is winning from $\rho_{k'}$. \RB{pas tr\`es clair...}\MS{J'ai reformul\'e; est-ce mieux?}
  So there is also a profile in $\stratset^n_{-i}$ that makes $\sigma_i$ win from $\rho''_k$.
  Namely:
  \[\sigma_j'''(h) = \left\{\begin{array}{ll}
  \rho_{\ell+1} & \textrm{if } h=\rho_{\leq \ell} \textrm{ with } \ell<k' \textrm{ and } \rho_{\ell} \in V_j \\
  \sigma''_j(h) & \textrm{otherwise}
  \end{array}\right.\]
\end{itemize}

  We can combine the strategy profile that makes $\sigma'_i$ lose with $\sigma'''_{-i}$ win to obtain one that is winning for $\sigma_i$ and losing for $\sigma'_i$.
  Hence $\sigma_i$ is not weakly dominated.
\qedhere
\end{proof}
\end{collect}

\begin{example}\label{ex:nonzero}
  As an illustration of Lemmata~\ref{lem:valun} and~\ref{lem:valmoinsun}, consider the left game in \figurename~\ref{fig:dj}.
  Both runs $s_0 \cdot s_1 \cdot s_2 \cdot \good_1^\omega$ and $s_0 \cdot s_1 \cdot s_2 \cdot s_3^\omega$ are outcomes of non dominated strategies of \player{1}.
  Indeed, the play that goes to $\good_1$ is winning and the value of the path that goes to $s_3$ belongs to $0^* (-1)^\omega$ and the value decreases on a transition by \player{2}.
\end{example}

\subparagraph{Value 0}
This case is more involved.
From a state of value $0$, an admissible strategy of \player{i} should allow a winning run for \player{i} with the help of other players.
We write $H^{n}_i$ for set of states~$s$ controlled by a player $j\ne i$ that have at least two successors that \emph{(i)}~have value $0$ or $1$ for \player{i} and \emph{(ii)}~have the same value for \player{j} than $s$ after iteration~$n-1$.\JFR{rewrite: not clear}
Formally, for $n\geq0$, the ``Help!''-states of player $i$ are defined as: \MS{J'ai ajout\'e la convention $\val^{-1}_j(s)=0$, peut \^etre qu'on peut le rappeler ici.}
\[H^n_i = \bigcup_{j\in\Agt\setminus\{i\}} \left\{ s \in V_j \mmid
\begin{array}{l}
\exists s',s'',\ s' \neq s'' \\
\wedge\ s\rightarrow s' \land s\rightarrow s'' \\
\wedge\ \val^{n}_i(s') \ge 0 \\
\wedge\ \val^{n}_i(s'') \ge 0 \\
\wedge\ \val^{n-1}_j(s) = \val^{n-1}_j(s')\\
\wedge\ \val^{n-1}_j(s) = \val^{n-1}_j(s'')
\end{array}\right\}\]
These states have the following property.
\begin{lemma}\label{lem:dj}
  Let $s\in V$, $i\in \Agt$ and $\rho \in \outcome_s(\stratset^n)$
  be such that $\hval{n}{i}(\rho) = 0^\omega$.
\begin{mathpar}
\rho \in \outcome_s(\stratset^{n+1}_i) \text{ if, and only if, } 
  \rho \in \win{i} \text{ or } \inff{\rho} \cap H^n_i \ne \emptyset
\end{mathpar}
\end{lemma}
\begin{proof}[Proof sketch]
  Assume first that there is only a finite number of visits to $H^n_i$ and that $\rho$ is not winning for \player{i}.
  Let~$k$ be the greatest index such that $\rho_k \in H^{n}_i$.
  After the prefix $\rho_{\le k}$ no profile of $\stratset^n$ can make $\sigma_i$ win, since no ``Help!''-state is visited.
  While since $\val^{n}_i(\rho_k)= 0$ there is a strategy profile $\profile''$ of $\stratset^n$ whose outcome is winning for $i$ from $\rho_k$.
  Hence the strategy that follows $\sigma_i$ until the prefix $\rho_{\le k}$ and then switch $\sigma''_i$ weakly dominates $\sigma_i$, and $\sigma_i \not\in \stratset^{n+1}_i$.

  \medskip

  Assume now that either $\rho\in \win{i}$ or there is an infinite number of indexes $k$ such that $\rho_k \in H^{n}_i$.
  We define a strategy $\sigma_i$ that follows $\rho$ when possible and otherwise plays a non-dominated strategy.
  Assume that there is a strategy profile $\sigma'_\Agt \in \stratset^n$ such that $\rho' = \outcome(\sigma'_{-i},\sigma'_i) \in \win{i}$ and $\rho'' = \outcome(\sigma'_{-i},\sigma_i) \notin \win{i}$.
  We show that there is a profile that makes $\sigma_i$ wins and $\sigma'_i$ lose, so that $\sigma'_i$ does not weakly dominate $\sigma_i$.

  For instance in the case where $\rho''$ continues to follow $\rho$ after it has diverged from $\rho'$, we know it will encounter a ``Help!''-state~$\rho_k$.
  From there there is a strategy profile $\sigma^d_{-i}$ that is $\stratset^n_{-i}$ thanks to the condition $\val^{n-1}_j(\rho_k) = \val^{n-1}(s'')$ whith $s''\ne \rho_{k+1}$, and which is winning for $\sigma_i$ thanks to the condition $\val^n_i(s'') \ge 0$, and the fact that $\sigma_i$ reverts to a non-dominated strategy outside of $\rho$.
  This strategy profile can moreover be made losing for $\sigma'_i$ from the point where it deviates from $\rho$, since this deviation is on a state controled by \player{i} and its value for \player{i} is $0$.
\end{proof}

\begin{collect}{appendix-prefix-independent}{\subsubsection{Proof of Lemma~\ref{lem:dj}}}{}
\begin{proof}
  Assume first that there is only a finite number of visits to $H^n_i$ and that $\rho$ is not winning for \player{i}.
  Let~$k$ be the greatest index such that $\rho_k \in H^{n}_i$.
  Let $\profile \in \stratset^n$ be a profile such that $\outcome(\profile)=\rho$.
  We will show that the strategy for \player{i} in this profile, namely $\sigma_i$, is weakly dominated with respect to $\stratset^{n}$.

  First, we show that after some prefix of $\rho$ no profile of $\stratset^n$ can make $\sigma_i$ win.
  Let $\sigmaa_{-i} \in \stratset^n_{-i}$ and consider the profile $\sigma'_{-i} = \sigma_{-i}\switch{\rho_{\le k+1}}{\sigmaa_{-i}}$ that follows $\profile$ until $\rho_{k+1}$.
  Let $\rho'$ the outcome of $(\sigma_i, \sigma'_{-i})$.
  If $\rho=\rho'$ then $\sigma'_{-i}$ does not make $\sigma_i$ win.

  Otherwise, consider $k'$ the first index where $\rho_{k'} \ne \rho'_{k'}$, and $j$ the player controlling state $s = \rho_{k'-1} = \rho'_{k'-1}$; we have that $j\ne i$. 
  Since $\sigma'_j \in \stratset^n_j$, Lemma~\ref{lem:admvaleq} yields that $\val^{n-1}_j(\rho'_{k'}) = \val^{n-1}_j(\rho'_{k'-1})$.
  Similarly, since $\sigma_j \in \stratset^n_j$, we also have $\val^{n-1}_j(\rho_{k'}) = \val^{n-1}_j(\rho_{k'-1})$.
  Therefore $\val^{n-1}_j(\rho_{k'}) = \val^{n-1}_j(\rho'_{k'}) = \val^{n-1}_j(s)$.
  Recall that $\val^{n}_i(\rho) = 0^\omega$ so $\val^{n}_i(\rho_{k'}) = 0$.
  Since $\rho_{k'-1} \not\in H^n_i$, it must be the case that $\val^{n}_i(\rho'_{k'}) = -1$.
  Therefore no strategy profile in $\prod_{j\ne i} \stratset^{n}_{j}$ makes $\sigma_i$ win from $\rho_{k+1}$.
  
  Secondly, since $\val^{n}_i(\rho_k)= 0$ there is a strategy profile $\profile''$ of $\stratset^n$ whose outcome is winning for $i$ from $\rho_k$.
  Hence $\sigma_i\switch{\rho_{\le k}}{\sigma''_i}$ weakly dominates $\sigma_i$, and $\sigma_i \not\in \stratset^{n+1}_i$.

  \medskip

  Assume now that either $\rho\in \win{i}$ or there is an infinite number of indexes $k$ such that $\rho_k \in H^{n}_i$.
  Let $\sigmaa_i$ be a strategy in $\stratset^{n+1}_i$.
  We build $\sigma_i$ so that it is compatible with $\rho$ and revert to $\sigmaa_i$ in case of a deviation.
  Formally, for a history~$h$, such that $\last(h)$ is controlled by \player{i}:
  \begin{itemize}
  \item if $h= \rho_{\le k}$ for some index $k$ then $\sigma_i(h) = \rho_{k+1}$ (this ensures that $\sigma_i$ is compatible with $\rho$);
  \item otherwise, let $k$ be the largest index such that $\rho_{\le k}$ is a prefix of $h$, then $\sigma_i(h) = \sigmaa_i(h_{\ge k+1})$.
  \end{itemize}
  
  We now show that $\sigma_i$ is not weakly dominated.
  Assume that there is a strategy profile $\sigma'_\Agt \in \stratset^n$ such that $\rho' = \outcome(\sigma'_{-i},\sigma'_i) \in \win{i}$ and $\rho'' = \outcome(\sigma'_{-i},\sigma_i) \notin \win{i}$.
  We show that there is a profile that makes $\sigma_i$ wins and $\sigma'_i$ lose, so that $\sigma'_i$ does not weakly dominate $\sigma_i$.
  
  We consider the points where $\rho''$ diverges from $\rho$ and $\rho'$, respectively.
  Formally, let $k$ (resp. $k'$) be the largest index such that $\rho''_{\leq k} = \rho_{\leq k}$ (resp. $\rho''_{\leq k'} = \rho'_{\leq k'}$).
  We distinguish two cases as illustrated in \figurename~\ref{fig:diverge}.
  \begin{figure}   
    \centering
    ~\hfill
    \subfloat[$k'\ge k$]{
      \label{fig:diverge1}
      \begin{tikzpicture}[scale=0.6]
        \draw (0,0) node (I) {};
        \draw (2,-2) node[inner sep=1pt,fill=black] (K) {};
        \draw (3,-3) node[inner sep=1pt,fill=black] (KP) {};
        \draw (0,-4) node (R) {$\rho$};
        \draw (1,-5) node (RP) {$\rho'$};
        \draw (5,-5) node (RPP) {$\rho''$};

        \draw (K.0) node[right] {$k$};
        \draw (KP.0) node[right] {$k'$};

        \draw (I) -- (K);
        \draw (K) -- (R);
        \draw (K) -- (KP);
        \draw (KP) -- (RP);
        \draw (KP) -- (RPP);
      \end{tikzpicture}
    }
    \hfill~\hfill
    \subfloat[$k' < k$]{
      \label{fig:diverge2}
      \begin{tikzpicture}[scale=0.6]
        \draw (0,0) node (I) {};
        \draw (2,-2) node[inner sep=1pt,fill=black] (KP) {};
        \draw (3,-3) node[inner sep=1pt,fill=black] (K) {};
        \draw (0,-4) node (RP) {$\rho'$};
        \draw (1,-5) node (RPP) {$\rho''$};
        \draw (5,-5) node (R) {$\rho$};

        \draw (K.0) node[right] {$k$};
        \draw (KP.0) node[right] {$k'$};

        \draw (I) -- (KP);
        \draw (KP) -- (K);
        \draw (KP) -- (RP);
        \draw (K) -- (R);
        \draw (K) -- (RPP);
      \end{tikzpicture}
    }
    \hfill~
    \caption{Divergence of outcomes by strategy shifting in the proof of Lemma~\ref{lem:dj}.}
    \label{fig:diverge}
  \end{figure}
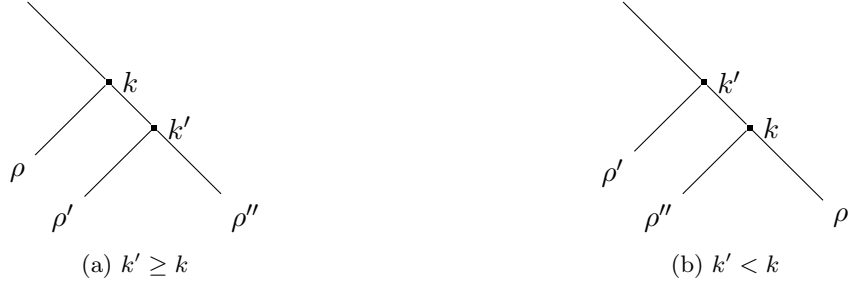
\begin{itemize}
\item If $k' \geq k$, \ie $\rho''$ diverges from $\rho$ before diverging from $\rho'$.
  Since $\sigmaa_i$ is not weakly dominated, from $\rho''_{k+1}$ there is a strategy profile $\sigma''_{-i}$ that makes $\sigmaa_i$ win and $\sigma'_i \circ \rho''_{\le k}$ lose.
  Hence $\sigma'_{-i}\switch{\rho''_{\le k'+1}}{\sigma''_{-i}}$ makes $\sigma_i$ win but not $\sigma'_i$.
  Therefore $\sigma'_i$ does not weakly dominates $\sigma_i$.

\item In the other case, $k'<k$, the point where the two outcomes diverge is along~$\rho$.
  We have $\rho_{\le k'} = \rho'_{\le k'}$.
  We write $s = \rho_{k'}' = \rho_{k'}$, which is controlled by $j \ne i$, $s' = \rho'_{k'+1}$, and $s''= \rho_{k'+1} = \sigma_i(\rho_{\le k'}) = \rho''_{k'+1}$.
  Notice that $\val^{n}_j(s')\le 0$ because $\val^n_i(s)=0$ and using Lemma~\ref{lem:value}.
  Hence there is a strategy profile~$\sigmal_{-i}$ in $\stratset^{n}$ that makes $\sigma'_i$ lose from $s'$.

\begin{itemize}
\item If $\rho$ is a winning outcome for \player{i}, we consider the strategy profile:
  \[\sigma'_{-i}\switch{\rho'_{\le k'+1}}{\sigmal_{-i} }\switch{\rho_{\le k'+1}}{\sigma_{-i}\circ \rho_{\le k'+1}} .\]
  Combined with $\sigma_i$, its outcome is $\rho$ and therefore is winning.
  This strategy profile belongs to $\stratset^{n}$ by Lemma~\ref{lem:shifting} and makes $\sigma_i$ win and $\sigma'_i$ lose, hence $\sigma'_i$ does not weakly dominates $\sigma_i$.

\item Otherwise, there is an infinite number of states in $H^n_i$ along $\rho$.
  Let $k''$ be the index of the first occurrence of $H^n_i$ after the divergence: the smallest $k'' > k'$ such that $\rho_{k''} \in H_i^n$.
  State $\rho_{k''}$ is controlled by a player $j\ne i$.
  We consider a strategy profile $\profile^d$ such that:
  \begin{itemize}
  \item $\sigma^d_j(\rho_{k''}) = s^d$ such that $s^d \neq \rho_{k''+1}$, with $\val^n_i(s) \ge 0$, and    $\val^{n-1}_j(\rho_{k''}) = \val^{n-1}_j(s^d)$.
    This state exists since $\rho_{k''} \in H^n_i$.
  \item Then from $s^d$ it follows a strategy profile $\sigma''_{-i} \in \stratset^n_{-i}$ that makes $\sigmaa_i$ win.
    This profile exists since $\sigmaa_i\in \stratset^{n+1}_i$ and $\val^{n}_i(s^d)\ge 0$. 
  \end{itemize}
  The profile $\profile^d$ belongs to $\stratset^n_{-i}$.
  Indeed, all players but $j$ play according to a strategy of $\stratset^n_{-i}$ by construction.
  In the case of \player{j}, it first involves a step that does not change the value $\val^{n-1}_j$, so by Lemma~\ref{lem:valeqadm}, it is also admissible.
  
  We consider the strategy profile:
  \[\sigma'_{-i}\switch{\rho'_{\le k'+1}}{\sigmal_{-i} }\switch{\rho_{\le k''}}{\sigma_{-i}^d} .\]
  This strategy profile belongs to $\stratset^{n}$ by Lemma~\ref{lem:shifting} and makes $\sigma_i$ win and $\sigma'_i$ lose, hence $\sigma'_i$ does not weakly dominates $\sigma_i$.
\end{itemize}
\end{itemize}
  
  This shows that $\sigma_i$ is not weakly dominated and belongs to $\stratset^{n+1}_i$.
\end{proof}
\end{collect}

\begin{figure}[htb]
  \begin{center}
    \subfloat[A first reachability game]{\label{fig:first-reach-game}
      \begin{tikzpicture}[xscale=2.4,yscale=1.5]
      \draw (0,0) node[player1] (A) {$s_0$};
      \draw (1,0) node[player1] (B) {$s_1$};
      \draw (2,0) node[player2] (C) {$s_2$};
      \draw (1.2,-1) node[draw,minimum size=1cm,rounded corners=5mm] (D) {$\good_1$};
      \draw (2.5,-1) node[player1] (E) {$s_3$};
      \draw[-latex'] (A) edge[bend left] (B);
      \draw[-latex'] (B) edge[bend left] (A);
      \draw[-latex'] (B) -- (C);
      \draw[-latex'] (C) -- (D);
      \draw[-latex'] (C) -- (E);
      \draw(E) edge[bigloop right,distance=0.3cm,-latex'] (E);
      \draw(D) edge[bigloop right,distance=0.3cm,-latex'] (D);
    \end{tikzpicture}
    }

    \subfloat[A second reachability game]{\label{fig:second-reach-game}
    \begin{tikzpicture}[xscale=2.4,yscale=1.5]
      \draw (0,0) node[player1] (A) {$s_0$};
      \draw (1,0) node[player1] (B) {$s_1$};
      \draw (2,0) node[player2] (C) {$s_2$};
      \draw (0.5,-0.5) node[player2] (CB) {$s_4$};
      \draw (1.2,-1) node[draw,minimum size=1cm,rounded corners=5mm] (D) {$\good_1$};
      \draw (2.5,-1) node[player1] (E) {$s_3$};
      \draw[-latex'] (A) edge[bend left] (B);
      \draw[-latex'] (B) edge[bend left] (CB);
      \draw[-latex'] (CB) edge[bend left] (A);
      \draw[-latex'] (B) -- (C);
      \draw[-latex'] (C) -- (D);
      \draw[-latex'] (C) -- (E);
      \draw[-latex'] (CB) -- (D);
      \draw (E) edge[bigloop right,distance=0.3cm,-latex'] (E);
      \draw (D) edge[bigloop right,distance=0.3cm,-latex'] (D);
    \end{tikzpicture}
    }
  \end{center}
    \caption{Two games. The goal for \player{1} is to reach $\good_1$. }
    \label{fig:dj}
\end{figure}
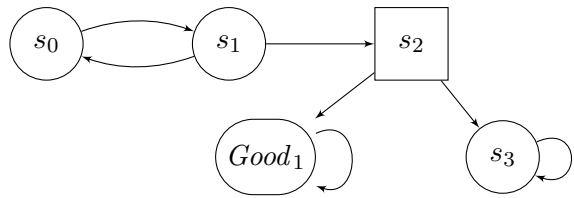
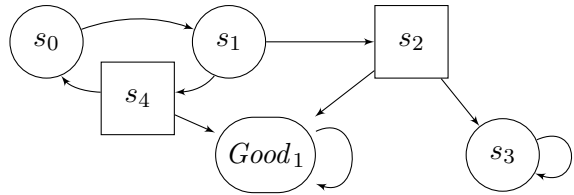

\begin{example}\label{ex:dj}
  As an illustration of Lemma~\ref{lem:dj}, consider the two games in \figurename~\ref{fig:dj}.
  In the game of \figurename~\ref{fig:first-reach-game}, a strategy of \player{1} that stays in the loop $(s_0\cdot s_1)$ forever is weakly dominated.
  The value of this run is $0^\omega$ and visits no ``Help!''-state, since in that game $H_1^0 = \emptyset$.
  Intuitively, the strategy is dominated because it has no chance of winning, while getting out after $k$ steps can be winning if \player{2} helps.

  However, in the game of \figurename~\ref{fig:second-reach-game}, the strategy that always chooses $s_4$ from $s_1$ is admissible.
  The run $(s_0 \cdot s_1 \cdot s_4)^\omega$ also has value $0^\omega$, but this time $H_1^0 = \{ s_4 \}$, which is visited infinitely often by the loop.
\end{example}

\subsection{Automata for $\outcome(\stratset^n)$}\label{sec:automataDef}
Our goal is to obtain an automaton which recognizes $\outcome(\stratset^n)$.
We decompose this construction for each player by noting that:
\[\outcome(\stratset^{n-1}) \cap \bigcap_{i\in \Agt} \outcome(\stratset^{n-1}_{-i},\stratset^n_i) = \outcome(\stratset^n) \]
From the 
characterization of admissible strategies w.r.t. values of state, we define an automaton~$\A_i^n$
that selects $\outcome(\stratset^{n}_{i})$ within $\outcome(\stratset^{n-1})$, \ie such that \(\L(\A_i^n) \cap \outcome (\stratset^{n-1})\) is equal to \(\outcome(\stratset^n_i,\stratset^{n-1}_{-i}) \).
This construction rely on values at previous iterations.
For $n>0$, $\A_i^n$ is the automaton where:
\begin{itemize}
\item The set of states is $V$, i.e. the same states as in $\G$;
\item From the transitions in $\G$ we remove those controlled by \player{i} that decrease his value.
  Formally:
  \[ T = E \setminus \{ (s,s') \mid s \in V_i \land \val^{n-1}_i(s) > \val^{n-1}_i(s') \}.\]
\item A run $\rho$ is accepted by $\A_i^n$ if one of the following conditions is satisfied:
\begin{enumerate}
\item $\hval{n-1}{i}(\rho) \in 0^* (-1)^\omega$; 
\item $\hval{n-1}{i}(\rho) \in 0^* 1^\omega$ and $\rho\in \win{i}$; 
\item $\hval{n-1}{i}(\rho) \in 0^\omega$ and $\rho\in \win{i}$ or $\rho \in (V^* H^{n-1}_i)^\omega$.
\end{enumerate}
\end{itemize}

Note that the structure of the automaton is a subgame of $\G$.
Since the transitions of $\G$ bear no label, the ``language''~$\L(\A)$ of an automaton~$\A$ is here considered to be the set of accepting runs.
The following lemma shows the relation between automaton $\A_i^n$ and outcomes of admissible strategies at step~$n$.

\begin{lemma}\label{lem:la}
  \( \outcome (\stratset^{n-1}) \cap \L(\A_i^n) = \outcome(\stratset^n_i,\stratset^{n-1}_{-i}) \)
\end{lemma}

\begin{collect}{appendix-prefix-independent}{\subsubsection{Proof of Lemma~\ref{lem:la}}}{}
\begin{proof}
  Let $\rho$ be the outcome of a strategy profile in $\stratset^{n-1}$ that is also accepted by $\A_i^n$.
  Note that proving that $\rho \in \outcome(\stratset^n_i)$ suffices to prove that $\rho \in \outcome(\stratset^n_i,\stratset^{n-1}_{-i})$.
  
  There are three different cases, depending on which part of the accepting condition of $\A^n_i$ is fulfilled.
  \begin{itemize}
  \item If $\hval{n-1}{i}(\rho) \in 0^* 1^\omega$, then $\rho\in\win{i}$, so Lemma~\ref{lem:valun} yields that $\rho \in \outcome(\stratset^n_i)$.
  \item If $\hval{n-1}{i}(\rho) \in 0^* (-1)^\omega$, the definition of $T$ ensures that for any index $k$ such that $\rho_k \in V_i$, $\hval{n-1}{i}(\rho_k) = \hval{n-1}{i}(\rho_{k+1})$.
  Thus Lemma~\ref{lem:valmoinsun} yields that $\rho \in \outcome(\stratset^n_i)$.
  \item Otherwise $\hval{n-1}{i}(\rho) \in 0^\omega$, and $\rho\in \win{i}$ or there is an infinite number of visits to $H_i^{n-1}$.
  Then Lemma~\ref{lem:dj} allows to conclude.
  \end{itemize}
  
\medskip

Reciprocally, let $\rho \in \outcome(\stratset^n_i,\stratset^{n-1}_{-i})$.
By monotonicity of the sets of admissible strategies, $\rho \in \outcome(\stratset^{n-1})$, so it remains to be proven that $\rho \in \L(\A^n_i)$.
Notice that since $\rho \in \outcome(\stratset^n)$, $\rho \notin B^n_i$, so $\rho$ is actually a run of $\A^n_i$.

Now by Lemma~\ref{lem:disjonctionval}, we can consider the three following cases for $\val^{n-1}_i(\rho)$.
\begin{itemize}
  \item If $\hval{n-1}i(\rho) \in 0^* (-1)^\omega$, then $\rho \in \L(\A^n_i)$.
  \item If $\hval{n-1}i(\rho) \in 0^* 1^\omega$.
  Let us write $\rho = \outcome(\sigma_i,\sigma_{-i})$ with $\sigma_i \in \stratset^n_i$ and $\sigma_{-i} \in \stratset^{n-1}_{-i}$.
  Let $s$ be the first state of value $1$ encountered in $\rho$.
  From $s$ there exists a strategy $\sigma^w_i$ for \player{i} winning against all other strategy profiles of $\stratset^{n-1}_{-i}$.
  Therefore if $\sigma_i$ is not a winning strategy, $\sigma^w_i \domstr{\stratset^{n-1}} \sigma_i$, which is a contradiction with the fact that $\sigma_i \in \stratset^n_i$.
  Thus $\sigma_i$ is a winning strategy so $\rho \in \win{i}$ and therefore $\rho \in \L(\A^n_i)$.
  \item Otherwise $\hval{n-1}i(\rho) \in 0^\omega$ and Lemma~\ref{lem:dj} allows to conclude that either $\rho\in \win{i}$ or $\rho \in (V^* H^{n-1}_i)^\omega$, therefore that $\rho \in \L(\A^n_i)$.
  \qedhere
\end{itemize}
\end{proof}
\end{collect}

The following property is a direct consequence of Lemma~\ref{lem:la}.
\begin{lemma}\label{lem:auto-stratsetn}
  \(\outcome(\stratset^{n-1}) \cap \bigcap_{i\in\Agt} \L(\A_i^n) = \outcome(\stratset^n) \)
\end{lemma}
\begin{collect}{appendix-prefix-independent}{\subsubsection{Proof of Lemma~\ref{lem:auto-stratsetn}}}{}
\begin{proof}
%
  \begin{eqnarray*}
    \outcome(\stratset^{n-1}) \cap \bigcap_{i\in\Agt} \L(\A_i^n) &=& \bigcap_{i\in\Agt} (\L(\A_i^n) \cap \outcome(\stratset^{n-1})) \\
    &=& \bigcap_{i\in\Agt} (\outcome(\stratset^n_i) \cap \outcome(\stratset^{n-1})) \\&& \hspace{4cm}\text{(by Lemma~\ref{lem:la})} \\
    &=& \left( \bigcap_{i\in\Agt} \outcome(\stratset^n_i)\right) \cap \outcome(\stratset^{n-1}) \\
    &=& \outcome(\stratset^n)  \cap  \outcome(\stratset^{n-1})\\
    \outcome(\stratset^{n-1}) \cap \bigcap_{i\in\Agt} \L(\A_i^n) &=& \outcome(\stratset^n) \ \text{because}\ \stratset^n \subseteq \stratset^{n-1}
  \end{eqnarray*}
  \qedhere
\end{proof}
\end{collect}

As $\outcome(\stratset^0)$ is the set of all runs over $\G$, the previous lemma allows the construction of $\outcome(\stratset^{n})$ by induction, for any $n$.
Note that the size of the automaton accepting $\outcome(\stratset^n)$ does not grow since all automata in the intersection share the same structure as $\G$ (although some edges have been removed).
The accepting condition, however, becomes more and more complex.
Nonetheless, if for all $j \in \Agt$, $\win{j}$ is a Muller condition given by a Boolean circuit, the condition of the automaton accepting $\outcome(\stratset^n)$ is a Boolean combination of such conditions, thus it is expressible by a Boolean circuit, of polynomial size.
For example, condition $\hval{n-1}{i}(\rho) \in 0^* 1^\omega$ can be expressed by a circuit that is the disjunction of all states of value $1$ at step $n-1$.
\MS{Je ne sais pas s'il faut plus de d\'etails l\`a dessus.}

%% file: automata.tex
\subsection{Inductive computation of values}
\label{sec:valcomput}

Now, we show how to compute the values of the states for all players.
Initially, at ``iteration $-1$'', all values are assumed to be $0$, and $\stratset^{-1}=\stratset$.
Computing the values at the next iteration relies on solving two-player zero-sum games with objectives based on outcomes of admissible strategies.
For example, in order to decide whether a state $s$ has value $1$ for \player{i} at iteration $0$, \ie whether $\val^0_i(s)=1$, one must decide whether \player{i} has a winning strategy from $s$ when playing against all other players.
This corresponds to the game with winning condition $\win{i}$ where vertices of all players but $i$ belong to a single opponent.
To decide whether $\val^0_i(s)>-1$, all players try together to make $i$ win.
Therefore this is a one-player game (or emptiness check) on $\G$ with condition $\win{i}$.
All the other states have value $-1$.

This idea is extended to subsequent iterations, but the objectives of the games become more complex in order to take into account the previous iterations: the objectives need to enforce that only outcomes of admissible strategies (at previous iterations) are played.
Hence the construction relies on the automata $\A^n_i$ built above.
Assuming that winning conditions for all players are Muller 
 conditions, this yields a polynomial space algorithm to compute the values.

\begin{collect}{appendix-prefix-independent}{}{}
In the sequel, a strategy or an outcome ``winning for objective $X$'' is said \emph{$X$-winning}, as not to confuse with simply \emph{winning}, meaning ``winning for objective $\win{i}$'' ($i$ being usually clear from context).
\end{collect}

\subsubsection{Characterizing states of value -1}\label{sec:psi}


A state~$s$ has a value~$>-1$ for \player{i} at step~$n$ if there is a strategy profile $\profile\in\stratset^n$ s.t. $\outcome_s(\profile) \in \win{i}$.
This is expressed by:
\(\exists \profile \in \stratset.\ \profile \in \stratset^n \land \outcome_s(\profile) \in \win{i}\).
This prompts the definition of the following objective: 
\[\Psi^n_i(s) = \outcome_s(\stratset^n) \cap \win{i}.\]

\begin{lemma}
If $\Psi^n_i(s) \neq \emptyset$, then $\exists \profile \in \stratset^n : \outcome_s(\profile) \in \win{i}$.
\end{lemma}

%

Therefore, since the set $\outcome_s(\stratset^n)$ and $\win{i}$ can both be expressed as the language of an automaton, testing whether a state has value $-1$ boils down to check emptiness.

\subsubsection{Characterizing states of value 1}

A state~$s$ has value~$1$ for \player{i} at step~$n$ if he has a strategy $\sigma_i$ in $\stratset^n_i$ such that for all strategies $\sigma_{-i}$ in $\stratset^n_{-i}$, $\outcome_s(\sigma_i,\sigma_{-i}) \in \win{i}$.
This is expressed by the formula:
\begin{mathpar}
\exists \sigma_i \in \stratset_i,\ \forall \sigma_{-i} \in \stratset_{-i},\ \left( \sigma_i\in \stratset^n_i \land \left(\sigma_{-i} \in \stratset^n_{-i} \Rightarrow \win{i}^s(\sigma_i,\sigma_{-i})\right)\right).
\end{mathpar}
This prompts the definition of the following objective:
\[\Omega_i^n(s)=\outcome_s(\stratset^{n}_i) \cap \left( \outcome_s(\stratset^n) \Rightarrow \win{i}\right).\]
We show that there is indeed a correspondence between this objective $\Omega_i^n$ and states~$s$ such that $\val^n_i(s)=1$.

\begin{collect}{appendix-prefix-independent}{}{}
\begin{lemma}\label{lem:omegaadm}
If $\sigma_i\in \stratset_i$ is a $\Omega_i^n(s)$-winning strategy, then there exists $\sigma_i' \in \stratset_i^n$ such that $\sigma_i(s\cdot h) = \sigma_i'(s\cdot h)$ for any history $h$.
\end{lemma}

\begin{proof}
  Let $\sigma_i^0 \in \stratset^n_i$.
  We show that for runs starting from $s$, $\sigma_i$ is admissible.
  That directly yields the result with $\sigma_i'=\sigma_i^0\switch{s}{\sigma_i}$.
  
  We only consider runs starting in state $s$.
  We show that $\sigma_i$ is not strictly dominated with respect to $\stratset^{n-1}$.
  Assume there exists a strategy profile $\profile' \in \stratset^{n-1}$ whose outcome~$\rho'$ is winning for \player{i}, but such that $\rho = \outcome_s(\sigma_i,\sigma'_{-i})$ is not winning for \player{i}.
  Let $k$ be the first index where $\rho_k \ne \rho'_k$.
  By hypothesis that $\sigma_i$ is $\Omega_i^n$-winning, $\rho \in \outcome_s(\stratset_i^n)$ and is the outcome of an admissible strategy $\sigmaa_i \in \stratset_i^n$.
  Note that $\rho = \outcome_s(\sigmaa_i,\sigma'_{-i})$ and is losing, while $\rho'=\outcome_s(\sigma_i',\sigma_{-i}')$ is winning, with $\sigma_{-i}' \in \stratset^{n-1}_{-i}$, so $\sigmaa_i \not\dom{\stratset^{n-1}} \sigma'_i$.
  By admissibility of $\sigmaa_i$, $\sigma_i' \not\domstr{\stratset^{n-1}} \sigmaa_i$, therefore $\sigma_i' \not\dom{\stratset^{n-1}} \sigmaa_i$.
  Hence there exists $\sigma_{-i}''\in \stratset^{n-1}$ that makes $\sigma'_i$ lose from $\rho'_k$.

  Now, $\sigma_i$ is winning for objective $\Omega_i^n$, and since $\rho_{\le k}$ is compatible with $\sigma_i$, it is $\Omega_i^n$-winning from $\rho_k$.
  In particular, if all players $j\ne i$ play according to a strategy of $\stratset^n_{-i}$ it is winning.
  Since the condition $V^\omega\setminus\outcome_s(\stratset^n_{-i})$ cannot be satisfied, condition $\win{i}$ is satisfied.
  Therefore the strategy profile $\sigma_{-i}''\in \stratset^{n-1}$ makes $\sigma_i$ win.
  Recall that $\sigma_{-i}''$ also makes $\sigma'_i$ lose from $\rho'_k$.
  Hence $\sigma_i$ is not strictly dominated with respect to $\stratset^{n-1}_i$ and it belongs to $\stratset^{n}_i$.
\qedhere
\end{proof}
\end{collect}

\begin{proposition}\label{prop:omega}
  A strategy of \player{i} is a strategy of $\stratset^{n}_i$ which is winning from state $s$ against all strategies of $\stratset^{n}_{-i}$ if, and only if, it is winning for objective~$\Omega_i^n(s)$.
\end{proposition}
\begin{collect}{appendix-prefix-independent}{\subsubsection{Proof of Proposition~\ref{prop:omega}}}{}
\begin{proof}
  Assume $\sigma_i$ is a strategy of $\stratset^n_i$ which is winning from $s$ against all strategies of $\stratset^n_{-i}$.
  Consider a strategy profile $\sigma_{-i}$ and the run they produce: $\rho=\outcome_s(\sigma_i,\sigma_{-i})$.
  Obviously $\rho\in \outcome_s(\stratset^n_i)$ because $\sigma_i\in \stratset^n_i$.
  If $\rho \in \outcome_s\left(\stratset^{n}_{-i}\right)$, as $\sigma_i$ is winning against the strategies of $\stratset^{n}_{-i}$ this means that $\rho$ is winning for \player{i}.
  Hence either $\rho\in\win{i}$ or $\rho \in V^\omega \setminus \outcome_s\left(\stratset^{n}_{-i}\right)$.
  So $\sigma_i$ wins for condition $\Omega_i^n$.

  \medskip

  Reciprocally, let $\sigma_i$ be winning for $\Omega_i^n(s)$.
  Lemma~\ref{lem:omegaadm} shows that $\sigma_i$ can be completed as a strategy of $\stratset_i^n$ when starting from states other than $s$.
  We now show that it is winning from $s$ against all strategies of $\stratset_{-i}^n$.
  Let $\sigma_{-i} \in \stratset^n_{-i}$ and $\rho=\outcome_s(\sigma_i,\sigma_{-i})$.
  The path $\rho$ is in $\outcome_s(S^n_{-i})$, so since $\sigma_i$ is winning for $\Omega^n_i$, it means that $\rho\in \win{i}$.
  Hence $\sigma_i$ is winning from $s$ against all strategies of $\stratset_{-i}^n$.
  \qedhere
\end{proof}
\end{collect}

Note that we cannot directly obtain a description of $\Omega_i^n(s)$ in term of automata using Lemma~\ref{lem:la}, since we need to recognize $\outcome_s(\stratset^n_i)$.
To characterize runs of $\outcome_s(\stratset^n_i)$ we will both use a condition on transition similar to the safety case, and a condition on the long run.
Recall that $T^n = \cup_{i\in\Agt} T^n_i$, where:
\[T^n_i =  \{ (s,s') \in E \mid s \in V_i \land \val^{n-1}_i(\rho_{k+1}) < \val^{n-1}_i(\rho_{k}) \}.\]

If a player $j \neq i$ takes a transition that decreases its value, then we immediately know that \player{i} wins for objective $\Omega^n_i$ by playing an admissible strategy.
We thus define:
\[ C^n_i = \{ \rho \in V^\omega \mid \exists k,\ \exists j\ne i,\ (\rho_k,\rho_{k+1}) \in T^n_j \};\]
We also define $C^{\leq n}_i = \bigcup_{m=1}^n C^m_i$ and $\nBC{\leq n} = \bigcap_{m=1}^n \nBC{n}$.
%
\begin{collect}{appendix-prefix-independent}{}{}
And similarly the sets $B^n_i$ and $B^{\leq n}_i=\bigcup_{m=1}^n B^m_i$ by:
\begin{align*}
  B^n_i &= \{ \rho \in V^\omega \mid \exists k.\ \rho_k \in V_i, \val^{n-1}_i(\rho_{k+1}) < \val^{n-1}_i(\rho_{k}) \};\\
\end{align*}
\end{collect}
%
\begin{collect}{appendix-prefix-independent}{}{}
\begin{lemma}\label{lem:admBC}
For any \player{i} and iteration $n$, the following implications hold:
\begin{enumerate}[label=(\arabic*)]
\item \label{item:admBunion}
  If $\rho \in \outcome(\stratset^n_i)$, then $\rho \notin B^{\leq n}_i$.
\item \label{item:admCunion}
  If $\rho \in \outcome(\stratset^n_{-i})$, then $\rho \notin C^{\leq n}_i$.
\item \label{item:admBCunion}
  If $\rho \in \outcome(\stratset^n)$, then $\rho \in \nBC{n}$.
\end{enumerate}
\end{lemma}

\begin{proof}
Since $\outcome(\stratset^{n-1}_i) \subseteq \outcome(\stratset^n_i)$, it is sufficient to prove the following.
\begin{enumerate}[label=(\arabic*')]
\item \label{item:admB}
  If $\rho \in \outcome(\stratset^n_i)$, then $\rho \notin B^n_i$.
\item \label{item:admC}
  If $\rho \in \outcome(\stratset^n_{-i})$, then $\rho \notin C^n_i$.
\end{enumerate}
  Implication~\ref{item:admB} holds because of Lemma~\ref{lem:admvaleq}.
  Implication~\ref{item:admC} is obtained by applying the same lemma to players $j \ne i$.
  Implication~\ref{item:admBCunion} is a consequence of~\ref{item:admBunion} and~\ref{item:admCunion}.
\qedhere
\end{proof}
\end{collect}
Objective $\Omega^n_i$ can 
be further decomposed with respect to~$C^n_i$ and $T^n$.
\begin{proposition}\label{prop:theta}
  Player~$i$, has a winning strategy for $\Omega^n_i(s)$ in $\G$ if, and only, if he has one for:
\[\Theta^n_i(s) = C^{\leq n}_i \cup \left( \nBC{\leq n} \cap \Omega^n_i(s)\right).\]
\end{proposition}

\begin{collect}{appendix-prefix-independent}{\subsubsection{Proof of Proposition~\ref{prop:theta}}}{}
\begin{proof}
  Objective $\Theta^n_i$ can be rewritten
\[\Theta^n_i(s) = C^{\leq n}_i \cup \left((V^\omega \setminus B^{\leq n}_i) \cap (V^\omega \setminus C^{\leq n}_i) \cap \Omega^n_i(s)\right).\]
Let $\sigma_i$ be a $\Omega^n_i(s)$-winning strategy.
Let $\rho \in \outcome_s(\sigma_i)$.
Note that since $\sigma_i$ is $\Omega^n_i(s)$-winning, player $i$ never decreases his value in $\rho$, therefore $\rho \in (V^\omega \setminus B^{\leq n}_i)$.
Therefore
\begin{itemize}
\item either $\rho \in C^{\leq n}_i \subseteq \Theta^n_i$,
\item or $\rho \in V^\omega \setminus C^{\leq n}_i$, then $\rho \in (V^\omega \setminus B^{\leq n}_i) \cap (V^\omega \setminus C^{\leq n}_i) \cap \Omega^n_i(s) \subseteq \Theta^n_i(s)$.
\end{itemize}
Therefore a $\sigma_i$ is also winning for $\Theta^n_i$.

\medskip

Reciprocally, if $\sigma_i$ is winning for $\Theta^n_i(s)$.
We build a $\Omega^n_i(s)$-winning strategy $\sigma_i'$ as follows.
Let $\sigmaa_i \in \stratset^n_i$ be an admissible strategy for player $i$.
We set $\sigma_i'$ that plays like $\sigma_i$ until a player $j \neq i$ decreases his value, then plays like $\sigmaa_i$.\MS{C'est un shifting g\'en\'eralis\'e...}

Note that $\sigma_i$ never decreases the value for $i$ before shifting to $\sigmaa_i$.
Indeed, assume against a profile $\sigma_{-i}$, \player{i} decreases his value after a history $h$ (where no other player has decreased his own value).
Let $\sigma'_{-i}$ be a (positional) profile that never decrease a value for $j \neq i$ at any iteration and  $\rho = \outcome_s(\sigma_i,\sigma_{-i}\switch{h}{\sigma_i'})$.
Then $\rho \notin C^{\leq n}_i$ while $\rho \in B^n_i \subseteq B^{\leq n}_i$, so $\rho \notin \Theta^n_i$, which is a contradiction with the fact that $\sigma_i$ is $\Theta^n_i$-winning.

As a result, applying Lemma~\ref{lem:valeqadm} several (but a finite number of) times yields that $\sigma_i'$ is admissible.
So every run produced by $\sigma_i'$ starting from $s$ is by definition in $\outcome_s(\stratset^n_i)$.
Also remark that $\sigma_i'$ is a $\Theta^n_i(s)$-winning strategy.

Let $\sigma_{-i}$ be a profile and let $\rho = \outcome_s(\sigma_{i}',\sigma_{-i})$.
If for some $j \neq i$, $\rho \notin \outcome_s(\stratset^n_j)$, then $\rho \in \Omega^n_i(s)$ because $\rho \notin \outcome_s(\stratset^n_{-i})$.
Otherwise, by Lemma~\ref{lem:admBC}, $\rho \notin C^{\leq n}_i$ thus $\rho \in \Omega^n_i(s)$.
\qedhere
\end{proof}
%
%
\end{collect}

\subsubsection{Using the automata}\label{sec:theta1}
We can now use the acceptance conditions of automata $(\A^m_j)_{m\leq n, j\in\Agt}$ to rewrite $\Theta_i^n$.
If we expand the definition of $\Omega^n_i$ in $\Theta^n_i$, we obtain that $\Theta^n_i(s)$ equals:
\begin{mathpar}C^{\leq n}_i \cup \left( \nBC{\leq n} \cap \outcome_s(\stratset^{n}_i) \cap \left(\outcome_s(\stratset^n) \Rightarrow \win{i}\right)\right).\end{mathpar}
The set $C^{\leq n}_i$ and $\win{i}$ are easily definable by an automaton.
An automata for $\outcome(\stratset^n)$ is constructed through an intersection of automata $(\A^m_j)_{m\leq n, j\in\Agt}$, as shown in Lemma~\ref{lem:auto-stratsetn}.
Hence, we now have to construct automata recognizing  and $\outcome(\stratset^n_i)\cap \nBC{\leq n}$.

\medskip

\subparagraph{Computing $\outcome(\stratset^n_i) \cap \nBC{\leq n}$}\label{sec:theta2}

This construction
is also based on the automaton $\A_i^n$.
For this, we show that it can be defined through a combination of $\outcome(\stratset^n_i,\stratset^{n-1}_{-i})$ and of outcomes of admissible strategies at the previous iteration.
%
Namely:
\begin{collect}{appendix-prefix-independent}{}{}
\begin{lemma}\label{lem:admNotBC}
  \begin{eqnarray*}
  \outcome(\stratset^n_i) \cap \nBC{\leq n}  &=& \outcome(\stratset^{n-1}_i) \cap \nBC{\leq n}
  \\&& \hspace{1em}
  \cap \left(\outcome(\stratset^{n-1}) \Rightarrow \outcome(\stratset^n_i,\stratset^{n-1}_{-i})\right).
  \end{eqnarray*}
\end{lemma}
\begin{proof}
  \framebox{$\subseteq$}
  The inclusion holds because $\stratset^n_i  \subseteq \stratset^{n-1}_i$ and \[\outcome(\stratset^n_i) \cap \outcome(\stratset^{n-1}) = \outcome(\stratset^n_i,\stratset^{n-1}_{-i}).\]
    
  \framebox{$\supseteq$}
  In the other direction, let $\rho$ be such that \[\rho\in \nBC{\le n} \cap \outcome(\stratset^{n-1}_i) \cap \left(\outcome(\stratset^{n-1}) \Rightarrow \outcome(\stratset^n_i,\stratset^{n-1}_{-i})\right).\]
  If $\rho$ is an outcome of $\stratset^{n-1}$ then it is an outcome of $(\stratset^n_i,\stratset^{n-1}_{-i})$ and therefore in $\outcome(\stratset^{n}_i)$.

  Otherwise, $\rho$ is in $\outcome(\stratset^{n-1}_i) \setminus \outcome(\stratset^{n-1})$.
  Let $\sigmaa_i$ be a strategy of $\stratset^n_i$, we define $\sigma_i$ that follows $\rho$ and revert to $\sigmaa_i$ if the history deviates from $\rho$.
  Formally, for a history $h$, such that $\last(h)\in V_i$:
  \begin{itemize}
  \item if there is $k$ such that $h = \rho_{\le k}$, then $\sigma_i(h) = \rho_{k+1}$; 
    this is to ensure that $\rho$ is an outcome of $\sigma_i$;
  \item otherwise let $k$ be the last index such that $\rho_{\le k}$ is a prefix of $h$, then $\sigma_i(h) = \sigmaa_i(h_{>k})$.
  \end{itemize}
  Note that $\outcome(\sigma_i) \cap B^n_i = \emptyset$.
  Indeed, the only outcome of $\sigma_i$ which might not be in $\outcome(\stratset^n_i)$ is $\rho$ which does not belong to $B^n_i$ by assumption.

  We show that $\sigma_i$ is in $\stratset^n_i$.
  Assume there is a strategy profile $\profile' \in \stratset^{n-1}$ whose outcome $\rho'$ is winning for \player{i} but such that $\rho''=\outcome(\sigma_i,\sigma'_{-i})$ is not.
  The path $\rho''$ is different from $\rho$ since it belongs to $\outcome(\stratset^{n-1})$.

  However, as noted before, $\rho'' \notin B^n_i$.
  Let $k$ (resp. $k'$) be the first index where $\rho''$ differs from $\rho$ (resp. $\rho'$).
  We distinguish two cases, once again as in \figurename~\ref{fig:diverge}.
  If $k' \ge k$, $\sigma_i$ is playing according to a non-dominated strategy from~$\rho_{\le k+1}$, hence it is not strictly dominated by $\sigma'_i$. 
  
  Otherwise, $k' < k$.
  Let $s=\rho''_{k'}$, it is controlled by \player{i}.
  Let $s'= \sigma'_i(\rho_{\le k'})$, and $s'' = \sigma_i(\rho_{\le k'})$.
  Note that since $\rho \notin B^n_i$, $\val^{n-1}_i(s)=\val^{n-1}_i(s'')$.
  Additionally, if $\val^{n-1}_i(s)>\val^{n-1}_i(s')$, then a profile that fares better for $\sigma_i$ than for $\sigma_i'$ can easily be found, so $\sigma_i' \not\domstr{\stratset^{n-1}} \sigma_i$.
  Therefore we assume $\val^{n-1}_i(s)=\val^{n-1}_i(s')$ and distinguish cases according to the three possible values.
  \begin{itemize}
  \item If $\val^{n-1}_i(s)=-1$, then it is not possible for $\profile'$ to be winning, hence a contradiction.
  \item If $\val^{n-1}_i(s)=1$.
    First assume that $\val^{n-1}_i(\rho''_{k+1}) = 1$. 
    Then, since $\sigma_i$ plays according to a strategy of $\stratset^n_i$ after $\rho''_{\ge k+1}$, $\rho''$ is a winning path.
    This contradicts our hypothesis.
    
    Otherwise, consider the first state of $\rho''_{\ge k'}$ that does not have value $1$:
    the smallest $k''>k'$ such $\val^{n-1}_i(\rho''_{k''}) <1$.
    As a result, and since $\rho'' \notin B^n_i$, $\rho''_{k''} = \rho_{k''} \in V_j$ for some $j \neq i$.
    Since $\rho\not\in C^n_i$, $\val^{n-1}_j(\rho''_{k''}) = \val^{n-1}_j(\rho''_{k''-1})$.
    By Lemma~\ref{lem:valeqadm}, from the strategy of $\stratset^{n-1}_{-i}$ from $\rho''_{k''}$ that makes \player{i} lose, we can construct another one that makes \player{i} lose from $\rho''_{k''-1}$.
    This contradicts the fact that $\val^{n-1}_i(\rho''_{k''-1}) = 1$.
  \item If $\val^{n-1}_i(s)=0$.
  Then one can find a profile that makes $\sigma_i'$ lose.
  First assume that $\val^{n-1}_i(\rho''_{k+1})\ge 0$.
  Then a profile can be found that makes $\sigma_i$ win, since $\sigma_i$ plays according to an admissible strategy from that point.
  
  Otherwise, consider the step in which a player $j \neq i$ decreases $\val^{n-1}_i$: the smallest $k''>k'$ such  $\val^{n-1}_i(\rho_{k''+1})=-1$.
  We have that $\rho_{k''} \in V_j$ for some $j \neq i$.
  Since $\val^{n-1}_i(\rho_{k''})=0$, there exists a successor state of $\rho_{k''}$ of the same value for \player{j} that is of value $\ge 0$ for \player{i}: 
  a state $q \neq \rho_{k''+1}$ such that $\rho_{k''} \rightarrow q$, $\val^{n-1}_i(q) \ge 0$, and $\val^{n-1}_j(\rho_{k''})=\val^{n-1}_j(q)$.
  Then from $q$ there is a profile $\sigma_{-i}'' \in \stratset^{n-1}_{-i}$ that makes $\sigma_i$ win, since $\sigma_i$ plays according to an admissible strategy from that point.
  Notice that the profile that plays like $\sigma_{-i}'$ up to $\rho_{k''}$ then goes to $q$ and shifts to $\sigma_{-i}''$ is still a profile of $\stratset^{n-1}_{-i}$.
  \end{itemize}
  Hence $\sigma_i$ is admissible and $\rho \in \outcome(\stratset^n_i)$, which concludes the proof.
\qedhere
\end{proof}
\end{collect}

\begin{lemma}\label{lem:admNotBCrec}
 $\outcome(\stratset^n_i) \cap \nBC{\leq n}$ is equal to:
 \\$
 \outcome(\stratset^{n-1}_i) \cap \nBC{\leq n-1}_i \cap \left(\outcome(\stratset^{n-1}) \Rightarrow \L(\A^n_i)\right)
 \cap \nBC{\leq n}.
 $
\end{lemma}
\begin{collect}{appendix-prefix-independent}{\subsubsection{Proof of Lemma~\ref{lem:admNotBCrec}}}{}
\begin{proof}
Remark that $B^{\leq n}_i = B^n_i \cup B^{\leq n-1}_i$ and $C^{\leq n}_i = C^n_i \cup C^{\leq n-1}_i$, so
\begin{eqnarray}
\nBC{n} &=& V^\omega \setminus (B^{\leq n}_i \cup C^{\leq n}_i) \notag\\
         &=& V^\omega \setminus (B^n_i \cup B^{\leq n-1}_i \cup C^n_i \cup C^{\leq n-1}_i) \notag\\
         &=& \left(V^\omega \setminus (B^{\leq n-1}_i \cup C^{\leq n-1}_i)\right) \cap \left(V^\omega \setminus B^n_i\right) \cap \left(V^\omega \setminus C^n_i\right) \notag\\
\nBC{n} &=& \nBC{n-1}_i \cap \left(V^\omega \setminus B^n_i\right) \cap \left(V^\omega \setminus C^n_i\right).\label{eq:nBCrec}
\end{eqnarray}
In addition, $\outcome(\stratset^n_i,\stratset^{n-1}_{-i}) = \outcome(\stratset^{n-1}) \cap \L(\A^n_i)$, so
\begin{eqnarray}
\outcome(\stratset^{n-1}) &\Rightarrow& \outcome(\stratset^n_i,\stratset^{n-1}_{-i}) = \outcome(\stratset^{n-1}) \Rightarrow \left(\outcome(\stratset^{n-1}) \cap \L(\A^n_i)\right) \notag\\
\outcome(\stratset^{n-1}) &\Rightarrow& \outcome(\stratset^n_i,\stratset^{n-1}_{-i}) = \outcome(\stratset^{n-1}) \Rightarrow \L(\A^n_i).\label{eq:AniRec}
\end{eqnarray}
Rewriting Lemma~\ref{lem:admNotBC} with Equations~\eqref{eq:nBCrec} and~\eqref{eq:AniRec} yields the result.
\qedhere
\end{proof}
\end{collect}

The previous lemma provides a recurrence relation to compute the intersection $\outcome(\stratset^n_i) \cap \nBC{n}$, with base case being all the runs of $\G$.

\subsection{Bounding the number of iteration phases}

We can show that $\val_i^n(s) = 1 \Rightarrow \val_i^{n+1}(s) = 1$ as winning strategies are admissible and $\stratset^{n+1}_{-i} \subseteq \stratset^n_{-i}$ implies that winning strategies remain winning.
Similarly, we can show $\val_i^n(s) = -1 \Rightarrow \val_i^{n+1}(s) = -1$. 
So the number of times that the value function changes is bounded by $|P|\cdot |V|$.
This allows us to bound the number of iterations necessary to reach the fix-point:
\begin{proposition}\label{prop:bounditer}
\(\stratset^* = \stratset^{|\Agt|\cdot|V|}\).
\end{proposition}

\begin{collect}{appendix-prefix-independent}{\subsubsection{Proof of Proposition~\ref{prop:bounditer}}}{}
\begin{proof}
Since $\stratset^{n+1} \subset \stratset^n$, a state that has value $\neq0$ at a given iteration cannot change its value at a subsequent iteration.
Hence there are at most $|\Agt|\cdot|V|$ changes of value during the iteration process.
Since admissible strategies depend only on the values (even through the ``Help!''-states), the strategy elimination stops when the value stops changing.
\end{proof}
\end{collect}

\subsection{Algorithm for Muller conditions}\label{sec:algo}

The above construction yields procedures to compute the values, and in turn to solve the winning coalition problem.
The algorithm works as follow, starting from $n=0$ and incrementing it at each loop:
\begin{enumerate}
\item compute an automaton representing $\outcome(\stratset^n)$ with the help of $\L(\A^n)$, as described in Section~\ref{sec:automataDef};
\item compute for each player~$i$ the objective $\Theta^n_i$ as explained in \ifShort\else Section~\ref{sec:theta1}\fi;
\item compute for each player~$i$ the objective $\Psi^n_i$ (see Section~\ref{sec:psi});
\item compute for each player~$i$ and state $s$, $\val^n_i(s)$: we will show how to do this in the remainder of this section;
\item compute for each player~$i$ the set $H^n_i$.
\end{enumerate}

\begin{proposition}
Checking whether $\val^n_i(s)=1$ is in \PSPACE.
\end{proposition}

\begin{proof}
We define the following two-player game $\G^n_i$ as follows.
$\G^n_i$ has the structure of $\G$, except that the edges corresponding to a player decreasing his own value at any previous iteration have been removed (thus avoiding $T^{\le n}$ syntactically).
The players are \eve and \adam, the states of \eve are $V_i$ and the states of \adam are $\bigcup_{j\neq i} V_j$.
The winning condition for \eve is $\Theta^n_i$, which is a Muller condition.

If \eve has a winning strategy starting from state~$s$, then $\val^n_i(s)=1$, as shown by Propositions~\ref{prop:omega} and~\ref{prop:theta}.
Recall that for Muller condition, deciding whether there exists a winning strategy for \eve is in \PSPACE~\cite{hunter07}.
\qedhere
\end{proof}

\begin{proposition}
Checking whether $\val^n_i(s)=-1$ is in $\coNP$.
\end{proposition}

\begin{proof}
This is equivalent to the emptiness of $\Psi^n_i$ when starting from state~$s$.
Again, $\Psi^n_i$ is a circuit condition on $\G$ where edges decreasing the value of their owner are removed.
Testing non-emptiness amounts to guessing a lasso path starting from $s$, and checking that the states visited infinitely often correspond to a set of $\mathcal{F}$.
This $\NP$ algorithm answers true when $\val^n_i(s) \ne -1$.
We therefore have a $\coNP$ algorithm.
\qedhere
\end{proof}

\medskip

\begin{proof}[\proofname\ of Theorem~\ref{thm:circuitPspace}]
Let $W,L$ be the set of players that must win and lose, respectively, in the instance of the winning coalition problem.
The winning coalition problem is thus equivalent to the non-emptiness of \[\Phi = \outcome(\stratset^*) \cap \bigcap_{i \in W} \win{i} \cap \bigcap_{i \in L} \neg\win{i}\] over $\G$, which is a Muller condition.
This check is done in \NP.\MS{M\^eme remarque qu'au dessus.}

This however requires the computation of the condition corresponding to $\outcome(\stratset^*)$, hence of the values for all intermediate iterations before the fix-point.
Each iteration means solving a Muller game, 
 which is done in polynomial space.
Moreover, by Proposition~\ref{prop:bounditer}, there are at most $|\Agt| \cdot |V|$ iterations.
Thus the overall complexity is in \PSPACE.
\qedhere
\end{proof}

%% file: buchi.tex
In this section, We assume that each winning condition $\win{i}$ is given by a B\"uchi set $\buchi_i$.
A careful analysis of condition $\Theta^n_i$ shows that it can be reformulated as a parity condition.
Hence computing the value of a state boils down to solving two-player parity games.
Parity games are known to be in $\UP \cap \coUP$, but the question whether a polynomial algorithm exists for parity games has been open for several years~\cite{emerson1991tree,Jur98}.
Our algorithm thus works in polynomial time with call to oracles in $\NP\cap\coNP$, hence is also in $\NP\cap\coNP$~\cite{brassard79}.
This idea is the basis for the proof of Theorem~\ref{thm:Buchi}, which is detailed in the remainder of the section.


\begin{definition}
  A \newdef{parity condition} is given by a coloring function $\chi: V \rightarrow \{0,\dots,M\}$ with $M\in \N$.
  Accepted runs with respect to the coloring functions are the ones where the maximal color visited infinitely often is even:
  \[\win{\chi} = \left\{\rho \mmid \max(\inff{\chi(\rho)}) \textrm{ is even}\right\}.\]
\end{definition}

\subsubsection{Expressing $\Theta_i$ as a parity condition}

First note that the acceptance condition of $\L(\A_i^n)$ can be expressed by the following B\"uchi set, assuming that we remove all the edges that decrease the value of \player{i}:
\begin{align*}
  K_i^n = & \{ s \mid \val^{n-1}_i(s) = -1 \} \cup \{ s \mid \val^{n-1}_i(s) = 1 \land s \in \buchi_i \} \\ 
  & \cup \{ s \mid \val^{n-1}_i(s) = 0 \land s \in (\buchi_i  \cup H_i^{n-1}) \}
\end{align*}
Note that $\buchi_i \subseteq K_i^n$ for every $n$.

$\outcome(\stratset^0)$ is accepted by the automaton with the same structure as $\G$ and where the B\"uchi set is $V$.
Since \(\outcome(\stratset^{n-1}) \cap \bigcap_{i\in\Agt} \L(\A_i^n) = \outcome(\stratset^n) \) by Lemma~\ref{lem:auto-stratsetn}, $\outcome(\stratset^n)$ is recognized by an automaton whose acceptance condition is a conjunction of $n\times |\Agt|$ B\"uchi conditions.
By taking a product of the game with an automaton of size $n \times |\Agt|$, a condition $\outcome(\stratset^{n})$ can be expressed as a single B\"uchi set that we write $D^n$.

Recall that condition \(\Theta^n_i\) is 
\[ C^{\leq n}_i \cup \left( \nBC{n-1} \cap \outcome(\stratset^{n}_i) \cap \left(\outcome(\stratset^n) \Rightarrow \win{i}\right)\right).\]
We isolate the prefix-independent part of $\Theta^n_i$ by defining
\[\Gamma^n_i=\outcome(\stratset^{n}_i) \cap \left(\outcome(\stratset^n) \Rightarrow \win{i}\right)\]
which, by unfolding the recurrence relation of Lemma~\ref{lem:admNotBCrec}, can be rewritten:
\begin{mathpar}
\Gamma^n_i = \bigcap_{m=0}^{n-1} \left(\outcome(\stratset^{m}) \Rightarrow \left(\L(\A_i^{m+1}) \cap \outcome(\stratset^{m})\right)\right) \cap \left(\outcome(\stratset^n) \Rightarrow \win{i}\right)
\end{mathpar}
Also remark that the ``prefix-dependent'' part of $\Theta^n_i$ can be expressed by either removing edges (to avoid sets $T_i^n$) or adding an additional winning state (for sets $C_i$).

Each set $\outcome(\stratset^{m})$ is expressed by the B\"uchi condition $D^m$.
Similarly, condition $\left(\L(\A_i^{m+1}) \cap \outcome(\stratset^{m})\right)$ is a conjunction of $m\times |\Agt|+1$ B\"uchi conditions.
Therefore, in order to consider all these different conditions on the same game, it can be assumed that $\G$ is synchronized with a product of size\footnote{Precisely $|\Agt|\cdot (n-1) \cdot (n-2) +n$.} $O(n^2 \cdot |\Agt|)$.
Then each condition $\left(\L(\A_i^{m+1}) \cap \outcome(\stratset^{m})\right)$ can be expressed as a single B\"uchi set that we write $E^{m+1}_i$.

Notice that we have a ``chain'' of inclusion of the various objectives.
For each index $m$ and \player{i}:
\begin{equation}
\outcome(\stratset^m)\cap\L(\A_i^{m+1}) \subseteq \outcome(\stratset^m) \subseteq \outcome(\stratset^{m-1})\cap\L(\A_i^m).\label{eq:chain}
\end{equation}
Hence, if the left operand of an implication \(\outcome(\stratset^{m}) \Rightarrow \left(\outcome(\stratset^m) \cap \L(\A_i^{m+1})\right)\) is satisfied, then all such implications are satisfied for indexes $k<m$.\MS{Donner un nom \`a cette implication?}
And if the right operand is satisfied, then it is the case for indexes $k \leq m$.
The condition $\Gamma^n_i$ can then be defined by the parity condition given by the following coloring function $\chi^n_i$:
\begin{itemize}
\item if $s\in F_i$ then $\chi^n_i(s) = 2n + 2$
\item otherwise $\chi^n_i(s)$ is the maximum value between $0$, \linebreak
$2\cdot \max\left\{m \mmid s \in E^m_i\right\}$,
and $2\cdot \max\left\{m \mmid s \in D^m\right\} + 1$,
with the convention that $\max\emptyset=-\infty$.
\end{itemize}


\begin{lemma}\label{lem:toparity}
  For all index $n$ and \player{i}, $\rho\in\Gamma^n_i$ if, and only if, $\rho$ satisfies the parity condition $\chi^n_i$.
\end{lemma}
\begin{collect}{appendix-prefix-independent}{\subsection{Proof of Lemma~\ref{lem:toparity}}}{}
\begin{proof}
  Assume $\rho \in  \Gamma^n_i$.
  Let $m\leq n$ be the largest index such that $\rho \in \outcome(\stratset^m)$.
  
  If $m = n$ then $\rho \in \win{i}$, therefore $F_i$ is visited infinitely often by $\rho$.
  The highest color visited infinitely often is thus $2n+2$ and $\rho$ satisfies condition $\chi^n_i$.

  If $m < n$ then $\rho$ is not in $\outcome(\stratset^{m+1})$ and therefore, for $m' > m$, not in $\outcome(\stratset^{m'})$.
  Hence $\rho$ does not visit infinitely often odd colors greater than $2m + 2$.
  Moreover, since $\rho \in \Gamma^n_i$, $\rho \in \outcome(\stratset^m)\cap\L(\A_i^{m+1})$ hence it visits infinitely often states of $E^{m+1}_i$, which have an even color of (at least) $2m+2$.
  This means that $\rho$ satisfies condition $\chi^n_i$.

  \medskip     

  Reciprocally, assume that $\rho$ satisfies $\chi^n_i$.
  Let $c$ be the highest color that is reached infinitely often.
  
  If $c=2n+2$ then $\buchi_i$ is visited infinitely often, therefore $\rho\in \win{i}$.
  Moreover, for every $m\leq n$, $F_i \subseteq K_i^m$ so $K_i^m$ is also visited infinitely often, hence $\rho \in \L(\A_i^{m+1}) \subseteq \outcome(\stratset^{m}) \Rightarrow \left(\outcome(\stratset^{m}) \cap \L(\A_i^{m+1})\right)$, therefore $\rho \in \Gamma^n_i$.

  Otherwise $c=2m+2$ with $m < n$.
  Then set $E^{m+1}_i$ is visited infinitely often while, for every $m'> m$, the set $D^{m'}$ is not.
  So for $m' > m$, $\rho$ is not in $\outcome(\stratset^{m'})$, hence all implications are satisfied of index $> m$.
  Note that it is also the case for $m'=n$, so $\rho \in \outcome(\stratset^n) \Rightarrow \win{i}$.
  
  Moreover, since $\rho \in \outcome(\stratset^m)\cap\L(\A_i^{m+1})$, for all indexes $m'\leq m$, Equation~\eqref{eq:chain} yields that $\rho \in \outcome(\stratset^{m'})\cap\L(\A_i^{m'+1}$.
  So the implication $\outcome(\stratset^{m'}) \Rightarrow \left(\outcome(\stratset^{m'})\cap\L(\A_i^{m'+1})\right)$ is satisfied.
  Therefore $\rho \in \Gamma^n_i$.
\qedhere
\end{proof}
\end{collect}

\subsubsection{Solving the winning coalition problem for B\"uchi objectives}

We showed that all objectives encountered in the computation of values boil down to parity objectives.
This yields the following proof:
\begin{proof}[\proofname\ of Theorem~\ref{thm:Buchi}]
By reformulating $\Theta^n_i$ as a parity condition, checking whether the value of a state is $1$ amounts to solving a two player game with parity condition.
This is known to be in $\UP \cap \coUP$, although suspected to be in~$\P$.
Similarly, condition $\Psi^n_i$ is a conjunction of B\"uchi objectives, hence solving its emptiness ---~which means deciding whether the value of a state is $-1$~--- is in~$\P$.

Let $n_0$ be the index where the sets of admissible strategies stabilize (we have $n_0 \leq |\Agt| \cdot |V|$).
Let $W,L$ be the set of players that should win and lose, respectively.
Solving the winning coalition problem amounts to solving emptiness of
\[\Phi = \outcome(\stratset^{n_0}) \cap \bigcap_{i \in W} \win{i} \cap \bigcap_{i \in L} \neg\win{i}.\]
As the conjunction of $n_0+|W|$ B\"uchi conditions and a coB\"uchi condition, it can be expressed as a parity condition with 3 colors.
\qedhere
\end{proof}


%% file: LTL.tex
\subsection{Model-checking under admissibility}\label{sec:ltl}
  This section is devoted to the proof of Theorem~\ref{thm:LTLmcAdm}.

  Given a \textsc{LTL} formula $\phi$, the model-checking under admissibility problem is equivalent to the emptiness of 
  $\Phi = \L(\lnot \phi) \cap \outcome(\stratset^*)$.
  It was proven in~\cite{sistla87} that the language $\L(\lnot \phi)$ can be represented by a B\"uchi deterministic\MS{Tu es s\^ur qu'il est d\'eterministe? Est-ce essentiel? (Je ne crois pas puisqu'on devine un chemin dedans.} automaton~$\mathcal{A}$ with set of states $Q$ such that:
  \begin{itemize}
  \item the size of each state of $\mathcal{A}$ is polynomial in $|\phi|$;
  \item it can be checked if a state is initial in space polynomial in $|\phi|$;
  \item it can be checked if a state is in $F$ in space polynomial in $|\phi|$;
  \item each transition of $\mathcal{A}$ can be checked in space polynomial in $|\phi|$.
  \end{itemize}
  
\begin{lemma}
  Let $\A$ be an automaton, then for any run~$\rho$ of the automaton, there exists $\pi$ and $\tau$ such that: 
  \begin{itemize}
  \item $\rho'=\pi \cdot \tau^\omega$ is a run of $\A$;
  \item $\occ{\rho'} = \occ{\rho}$;
  \item $\inff{\rho'} = \inff{\rho}$;
  \item $|\pi|+|\tau| \le 2 \cdot n^2$ where $n$ is the number of states of $\A$.
  \end{itemize}
\end{lemma}

The proof can be seen as a particular case of Lemma~2.2 in~\cite{brenguier-phd2012}.
\begin{proof}
  We inductively construct a history $\pi = \pi_0 \pi_1
  \dots \pi_n$ that is not too long and 
  visits precisely those states that are visited by~$\rho$.

  The initial state is $\pi_0 = \rho_0$. 
  Then assume we have constructed $\pi_{\le k} = \pi_0 \dots \pi_k$ which visits exactly the same states as $\rho_{\le k'}$ for some~$k'$. 
  If all the states of~$\rho$ have been visited in $\pi_{\le k}$ then the construction is over. 
  Otherwise there is an index~$i$ such that $\rho_{i}$ does not appear in~$\pi_{\le k}$.
  The smallest such~$i$ is therefore the next \emph{target}: define $t(\pi_{\le k}) = \min \{ i \mid \forall j \le k.\ \pi_j \neq \rho_i\}$.
  Now consider the occurrence of the current state~$\pi_k$ that is the closest to the target in~$\rho$: $c(\pi_{\le k}) = \max \{ i < t(\pi_{\le k}) \mid \pi_k =\rho_i\}$.
  Run $\pi$ emulates what happens at that position by choosing $\pi_{k+1} = \rho_{c(\pi_{\le k})+1}$.  \MS{$i \rightarrow k$.}
  Then $\pi_{k+1}$ is either the target $\rho_{t(\pi_{\le k})}$, or a state that has already been seen before in $\pi_{\le k}$, in which case the resulting $\pi_{\le k+1}$ visits exactly the same states as $\rho_{\le c(\pi_{\le k})+1}$.

  At each step, either the number $|\occ{\rho}\setminus\occ{\pi_{\leq k}}|$ of remaining targets strictly decreases, or the number of remaining targets is constant but the distance $t(\pi_{\le k})-c(\pi_{\le k})$\MS{C'est bien \c ca?} to the next target strictly decreases.
  Therefore the construction terminates. 
  Moreover, notice that between two targets the same state is never visited twice, and only states that have already been visited, or the target, are visited.
  As the number of targets is bounded by~$n$, we obtain that the length of the path $\pi$ is bounded by $1+|n|\cdot (|n|-1)/2$.
  In addition, states of $\pi$ are always picked from $\rho$, ensuring that at each step $\occ{\pi_{\leq k}} \subseteq \occ{\rho}$.

  \smallskip
  Using similar ideas, we now inductively  construct $\tau = \tau_0 \tau_1 \dots \tau_m$,  which visits precisely those states which are seen infinitely often along~$\rho$, and which is not too long.  
  Let $\ell$ be the least index after which  the states visited by~$\rho$ are visited infinitely often: $\ell = \min\{ i \mid \forall j\geq i.\ \rho_j \in \inff{\rho}\}$.
  The~run~$\rho_{\ge \ell}$ is such that $\inff{\rho_{\geq \ell}} = \inff{\rho}$.
  Run~$\tau$ can be built in the same way as~$\pi$ above, but for play~$\rho_{\ge \ell}$. 
  As~a by-product, we also get $c(\tau_{\leq k})$, for~$k<m$.\MS{Je ne comprends pas cette phrase; en fait je ne comprends pas \`a quoi sert $c(\tau_{\leq k})$ hors de cette construction.}

  For $\pi \cdot \tau^\omega$ to be a real run, $\pi$  must be joined to $\tau$, and $\tau$ must be joined to itself.
  The last state of~$\pi$ must be the first state of~$\tau$ (and similarly the last state of $\tau$ must also be its first state).
  This possibly requires appending some more states to~$\pi$ and~$\tau$: the target of~$\pi$ and $\tau$ is set to be~$\tau_0$, and the same construction as previously is applied.
  As before, the added states belong to $\rho$ (resp. $\rho_{\geq \ell}$), ensuring that $\occ{\pi} = \occ{\rho}$ (resp. $\inff{\tau} = \inff{\rho}$).
  The total length of the resulting paths~$\pi$ and~$\tau$ is bounded by $1+(n - 1)\cdot(n+2)/2$ which less than~$n^2$.
\end{proof}
  
\begin{lemma}\label{lem:modelcheck}
  There is a word in $\Phi$ if and only if there are histories~$\pi$ and $\tau$ of length bounded by $(|\G| \times|\A|)^2$, such that $\pi\cdot \tau^\omega$ is a path in the product $\G \times \A$, $\tau^\omega$ satisfies the circuit condition defining $\outcome(\stratset^*)$, and $\tau$ visits an accepting state of $\A$.
\end{lemma}
\begin{proof}
By applying the preceding lemma to the product automaton $\G \times A$.
\end{proof}

From this lemma we deduce the non-deterministic polynomial space Algorithm~\ref{alg:modelcheck} to decide the model-checking under admissibility problem.

\begin{algorithm}[h]
\textsc{MCunderADM}$(\G_0,\phi)$::= \Begin{
Compute $(\G,\psi)$ as automaton with circuit condition for $\outcome(\stratset^*(\G_0))$\; \tcp*[f]{see Section~\ref{sec:automataDef}.}\\
Compute the machine representing the B\"uchi automaton $\A$ accepting $\L(\neg\phi)$\;
$\ell:=0$\;
Nondeterministically set $s \in V$ and $q \in Q$ an initial state of $\A$\;
$\textit{In}\pi := \textit{true}$\;
\While{$\textit{In}\pi \wedge \ell \leq (|\G| \times|\A|)^2$}{
Guess $s'$ such that $s \rightarrow s'$;
$s:= s'$\;
Guess $q'$ such that $q \xrightarrow{s'} q'$;
$q := q'$\;
Guess a boolean value for $\textit{In}\pi$;
$\ell := \ell+1$\;
}
\lIf{$\textit{In}\pi$}{\Return{false} \tcp*{$\ell$ reached its bound}}
\Else{
$t:=s$; $p:=q$; $\ell:=0$; $X:=\emptyset$; $b:=\textit{false}$\;
\While{$\neg((s,q) = (t,p) \wedge X \vDash \psi \wedge b) \wedge \ell \leq (|\G| \times|\A|)^2$}{
Guess $s'$ such that $s \rightarrow s'$;
$s:= s'$;
$X := X \cup \{s'\}$\;
Guess $q'$ such that $q \xrightarrow{s'} q'$;
$q := q'$;
$b := b \vee (q' \in F)$\;
$\ell := \ell+1$\;
}
\lIf{$(s,q) = (t,p) \wedge X \vDash \psi \wedge b$}{\Return{true}\;}
\lElse{\Return{false} \tcp*{$\ell$ reached its bound}}
}
}
\caption{\PSPACE\ algorithm for the model-checking under admissibility problem.}
\label{alg:modelcheck}
\end{algorithm}

The algorithm looks for $\pi$ and $\tau$ satisfying the conditions of Lemma~\ref{lem:modelcheck}.
It remembers the length~$\ell$ of the path guessed so far.
This length is bounded\footnote{If the bound is reached, the algorithm rejects the run.} by $(|\G| \times|\A|)^2$, which is exponential in the input, hence can be stored in polynomial space.
The algorithm also stores the current states of both $\G$ and $A$ ($s$ and $q$, respectively) and the last states of $\G$ and $A$ at the end of $\pi$ --~therefore the beginning of $\tau$~-- ($t$ and $p$, respectively).
Since there is only an exponential number of states in $\A$ (and a polynomial one in $\G$), coding a state can be done in polynomial space.
In order to check that the winning conditions are satisfied, the algorithm remembers the set~$X$ of states of $\G$ that have been visited (by $\tau$), and the Boolean~$b$ telling whether an accepting state of $\A$ as been visited (again, by $\tau$).
In the end, the fact that $\tau$ is actually a looping path is also checked by testing $(s,q)=(t,p)$.
Note that the paths $\pi$ and $\tau$ in $\G$ must respect the local conditions (\emph{e.g.} a player does not take an edge that decreases his value).

The correctness of the algorithm follows by Lemma~\ref{lem:modelcheck}, it terminates and use only polynomial space.
Hence the model-checking under admissibility problem is in \PSPACE.

\PSPACE-hardness derives from the \PSPACE-hardness of \textsc{LTL} model-checking~\cite{sistla85}.
This concludes the proof of Theorem~\ref{thm:LTLmcAdm}.

%% file: train.tex
\section{Example: a metro system}\label{sec:train}

\subsection{The model of the system}
We illustrate the use of model-checking under admissibility over the following metro system.
The metro track is composed of a directed ring divided in $n$ slots (one can assume for example that even-numbered slots represent a station, while odd-numbered slots represent a section of track between two stations).
There are $p$ trains numbered $1$ to $p$ on the track, initially at positions $p-1,\ldots,0$ (hence train number $1$ starts in front of the others (see \figurename~\ref{fig:traintrack}).
In the sequel we study the case of a track of length $6$ with two trains $t_1$ and $t_2$.
\begin{figure}[h!]
\centering
\begin{tikzpicture}
\def\dashdist{5pt}
\path[draw,thick,double,double distance=5pt] (0,0) arc (-90:90:1.5) -- ++(-5,0) arc (90:270:1.5) -- cycle;
\foreach \x in {0,-2.5,-5} {\path[draw,very thick] (\x,+\dashdist) -- (\x,-\dashdist);}
\begin{scope}[yshift=3cm]\foreach \x in {0,-2.5,-5} {\path[draw,very thick] (\x,+\dashdist) -- (\x,-\dashdist);}\end{scope}
\foreach \x/\sec in {-3.75/0,-1.25/1} {\node[anchor=north] (lab\sec) at (\x,-0.1) {$\sec$}; \node (sec\sec) at (\x,-0.25) {};}
\foreach \x/\sec in {-3.75/4,-1.25/3} {\node[anchor=north] (lab\sec) at (\x,2.9) {$\sec$}; \node (sec\sec) at (\x,2.75) {};}
\node[anchor=east] (lab2) at (1.4,1.5) {$2$}; \node (sec2) at (1.5,1.5) {};
\node[anchor=west] (lab5) at (-6.4,1.5) {$5$}; \node (sec5) at (-6.5,1.5) {};
\node[anchor=south] (tr1) at (sec1) {\minitrain{black}{white}{1}};
\node[anchor=south] (tr2) at (sec0) {\minitrain{lipicsyellow}{black}{2}};
\end{tikzpicture}
\caption{Metro track with $n=6$ and $p=2$, at their initial position $(1,0)$.}
\label{fig:traintrack}
\end{figure}
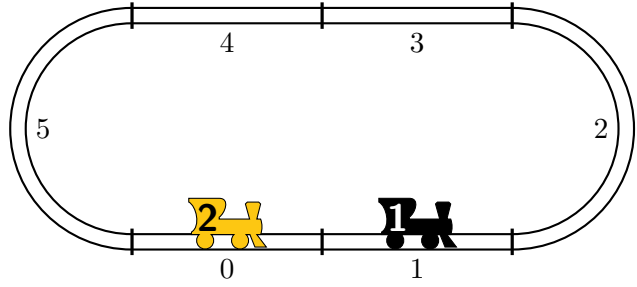

At each step, all the trains successively declare whether they want to advance or not.
Once everyone has chosen, all trains try to move synchronously.
However, there is an evil environment $\env$ that can prevent trains from moving (for example by activating an alarm on the section the train is trying to move in).
That means some trains that wanted to move may in fact remain in their position.
Nevertheless, all other trains must comply with their original choice, modeling the fact that trains cannot communicate\footnote{Although some kind of communication permits the trains to make their choice knowing in which sections the other trains are and what choice has been made by the trains with lower indices.} in real-time.
Additionally, if at any point two trains are on the same track section, they collide and the game stops.
(Part of) the game graph for this protocol is depicted in \figurename~\ref{fig:metrogame}.
In this game, the configuration of the track is given by the pair of positions for $t_1$ and $t_2$, and their respective wish to move or stay in place by the symbols ``$\rightarrow$'' (advance) or ``$\downarrow$'' (stay in place).

\begin{figure*}
\centering
\begin{tikzpicture}[auto,node distance=0.25cm and 2.25cm,>=latex']
\tikzstyle{player1}=[state,shape=rounded rectangle]
\tikzstyle{eliminated1}=[dotted,thick]
\tikzstyle{eliminated2}=[dashed, dash pattern=on 3pt off 3pt on 1pt off 3pt]
\tikzstyle{envstate}=[state,player3,text width=1.25cm,text centered,inner xsep=-1pt,inner ysep=0pt]
\node[state,player1] (q0) {$(1,0),(\bot,\bot)$};
\node[state,player2,node distance=1cm and 2cm,above right=of q0] (q1) {$(1,0),(\rightarrow,\bot)$};
\node[state,player2,node distance=1cm and 2cm,below right=of q0] (q2) {$(1,0),(\downarrow,\bot)$};
\node[envstate,above right=of q1] (q3) {$(1,0)$\\$(\rightarrow,\rightarrow)$};
\node[envstate,player3,below right=of q1] (q4) {$(1,0)$\\$(\rightarrow,\downarrow)$};
\node[envstate,player3,above right=of q2] (q5) {$(1,0)$\\$(\downarrow,\rightarrow)$};
\node[envstate,player3,below right=of q2] (q6) {$(1,0)$\\$(\downarrow,\downarrow)$};

\node[state,player1] (q7) at (12,2cm) {$(2,1)$};
\node[state,player1] (q8) at (12,0) {$(2,0)$};
\node[state,player1] (coll) at (12,-2cm) {$(1,1)$};

\path[->] (q0) edge node[swap,sloped,pos=0.33]{move} (q1);
\path[->] (q0) edge node[sloped,pos=0.333]{stay} (q2);
\path[->,eliminated2] (q1) edge node[swap,sloped,pos=0.1]{move} (q3);
\path[->] (q1) edge node[sloped,pos=0.33]{stay} (q4);
\path[->,eliminated2] (q2) edge node[swap,sloped,pos=0.1]{move} (q5);
\path[->] (q2) edge node[sloped,pos=0.33]{stay} (q6);

\path[->,eliminated1] (q3) edge[bend right] node[swap,sloped,pos=0.33]{block both} (q0);
\path[->,eliminated1] (q3) edge[bend left] (q7);
\path[->,eliminated1] (q3) edge node[sloped,pos=0.5]{block $2$} (q8);
\path[->] (q3) edge[bend left=10] node[sloped,pos=0.75]{block $1$} (coll);

\path[->] (q4) edge node[swap,sloped,pos=0.33]{block $1$} (q0);
\path[->] (q4) edge[bend right=5] (q8);
\path[->,eliminated1] (q5) edge node[swap,sloped,pos=0.33]{block $2$} (q0);
\path[->] (q5) edge (coll);
\path[->] (q6) edge[bend left] (q0);
\path[->] (coll) edge[loop below] (coll);

\tikzstyle{fading edge}=[dashed, dash pattern=on 7pt off 2pt on 2pt off 2pt on 2pt off 2pt on 2pt]
\path[draw,->,fading edge] (q7.north east) -- ++(0.33,0.75);
\path[draw,->,fading edge] (q8.north east) -- ++(0.33,0.75);
\path[draw,<-,fading edge] (q0.north) -- ++(-0.33,0.75);
\end{tikzpicture}
\caption[Part of the game played by the trains and the environment when trying to move from the original position.]{Part of the game played by the trains and the environment when trying to move from the original position. Round states belong to train $t_1$, square ones to train $t_2$, and triangular ones to the environment. Dotted edges are eliminated in the first iteration. Dashed ones in the second iteration.}
\label{fig:metrogame}
\end{figure*}
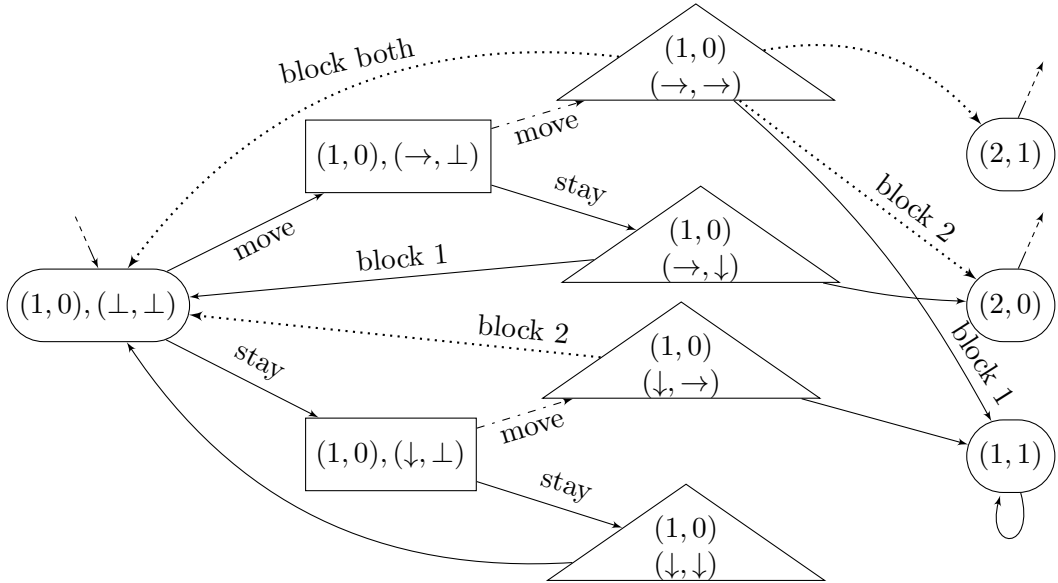

The objective of both trains is to loop infinitely often, which can be expressed by a generalized B\"uchi winning condition\footnote{This condition can therefore be expressed by a B\"uchi condition on a graph with memory of visits to $0$ or $1$. It can also be expressed by a circuit condition encoding formula $\bigvee (0,\_) \wedge \bigvee (1,\_)$.} requiring visiting infinitely often both sections $0$ and $1$.
The objective of the environment is to give rise to a collision: since collision are sink states it is equivalent to a B\"uchi winning condition over the set of collision states.
Globally, we are interested in knowing whether the \textsf{LTL} formula $\psi_{\neg\textrm{coll}}$ that expresses that no collision ever happens is satisfied under admissibility.
Remark that here:
\begin{itemize}
\item no player has a winning strategy alone;
\item the formula $\psi_{\neg\textrm{coll}}$ is not verified on all paths of the model;
\item that the players are not \emph{a priori} trying to satisfy $\psi_{\neg\textrm{coll}}$: indeed, the environment's objective is to negate $\psi_{\neg\textrm{coll}}$.
\end{itemize}

In the sequel we will show that the mechanism of admissibility enforces the satisfaction of $\psi_{\neg\textrm{coll}}$, although each train cannot guarantee that its objective will be fulfilled and $\psi_{\neg\textrm{coll}}$ is not an explicit requirement in the objective of any player.

\subsection{Strategy elimination}

\subsubsection{Value computation for first iteration}
At the first iteration, players for both trains never have a winning strategy for their winning condition that requires infinitely many visits to both positions $0$ and $1$: they cannot prevent, for example, the environment from always blocking their movement.
Note that one train can also make the other one lose on its own, by choosing to always stay in the same position.
On the other hand, in all states but the one with a collision, \emph{e.g.} state $(1,1)$, both trains may achieve their goal if the other players cooperate.
Hence the value of states for both train players is $0$ except in collision states, where it is $-1$.

For the environment player, however, there are states where he can force a collision, and hence has a winning strategy: whenever train~$2$ is right behind train~$1$ and train~$2$ has decided to move, which happens for example in states \mbox{$(1,0),(\rightarrow,\rightarrow)$} and $(1,0),(\downarrow,\rightarrow)$ in \figurename~\ref{fig:metrogame}.
These states therefore have value $1$ for the $\env$.
Additionally, from collision states the environment has already won, hence their value is also $1$.
From other states, if trains chose to never move, for example, $\env$ will not achieve its goal.
And they can also chose to reach a winning state for $\env$, hence the value from other states is $0$.

\subsubsection{Strategy elimination for first iteration}
From the values computed above, we first see that $\env$ eliminates all strategies that do not go to a collision state whenever it is possible: this amounts to removing some edge that originate from winning state, for example the dotted edges stemming from states $(1,0),(\rightarrow,\rightarrow)$ and $(1,0),(\downarrow,\rightarrow)$ in \figurename~\ref{fig:metrogame}.

Note that the strategy of $\env$ that always blocks train~$1$ if he is the only train wishing to move (in state $(1,0),(\rightarrow,\downarrow)$), which may end up losing for $\env$ if both players repeat their choice, is not eliminated, since this play traverses infinitely often state $(1,0),(\rightarrow,\bot)$, which is a ``Help!''-state for player $\env$.

Players $t_1$ eliminates for example the strategy that always stays in the same position, because it is always losing with respect to his objective.
However, although it is intuitively not a good choice for $t_1$ to always choose to ``stay'' in position $1$, the play \begin{mathpar}
\big((0,1),(\bot,\bot) \xrightarrow{stay} (1,0),(\downarrow,\bot) \xrightarrow{stay} (1,0),(\downarrow,\downarrow) \longrightarrow(0,1),(\bot,\bot)\big)^\omega
\end{mathpar} is an outcome of an admissible strategy.
For example, consider the strategy of $t_1$ that tries to move only if the path
\begin{mathpar}(0,1),(\bot,\bot) \xrightarrow{stay} (1,0),(\downarrow,\bot) \xrightarrow{move} (1,0),(\downarrow,\rightarrow) \xrightarrow{block 2} (0,1),(\bot,\bot)\end{mathpar} is taken first\footnote{One can show that this strategy is indeed admissible.}.
Against a strategy of $t_2$ that always chooses not to move, the aforementioned loop is the outcome of the game.
This is a consequence of the fact that $(1,0),(\downarrow,\bot)$ is a ``Help!''-state for player $t_1$.

Similarly, all states of $t_1$ are ``Help!''-states for player $t_2$.

As a result, the outcomes of strategies admissible after the first phase of elimination are all paths that never use the dotted transitions of player $\env$.

\subsubsection{Second iteration}

At the second iteration, the value for player $\env$ do not change.
On the other hand, the states where player $\env$ can force a collision are now losing for players $t_1$ and $t_2$, and hence now have value $-1$.

As a result, player $t_2$ will never play the transitions that lead to these states (the dashed transitions in \figurename~\ref{fig:metrogame}).
However, the ``Help!''-states for player $t_1$ and $t_2$ are the same, hence the outcomes of strategies admissible after two iterations remains unchanged (except for these local conditions).

\subsubsection{Third and fourth iteration}

At these iterations the values do not change.
Nonetheless, some ``Help!''-states have to be removed.
For example, state $(1,0),(\downarrow,\bot)$ belonging to $t_2$ is not a ``Help!''-state of $t_1$ anymore.
Hence any admissible strategy must, from $(0,1),(\bot,\bot)$, try to move infinitely often.

Neither values nor sets of admissible strategies change at the fourth iteration, meaning the set $\outcome(\stratset^*)$ has been computed.

\subsubsection{Analysis of $\outcome(\stratset^*)$}
In this set, one can see that whenever $t_2$ is right behind $t_1$:
\begin{itemize}
\item $t_2$ never tries to move;
\item $t_1$ infinitely tries to move;
\end{itemize}
As a result, one can see there can be no collision, hence the objective $\psi_{\neg\textrm{coll}}$ is satisfied.
Hence, the model-checking under admissibility problem answers positively for $\psi_{\neg\textrm{coll}}$, while the (``classical'') model-checking problem for the same formula answers negatively.

This also shows for example that the winning coalition problem with $\env \in W$ answers negatively: there is no admissible profile that makes player $\env$ win.

On the other hand, it cannot be assured that $t_1$ and $t_2$ will win for their objective: the environment can choose to always block their movement, preventing them from looping infinitely often on the track.

%% file: reachability.tex


In this section, we assume that the winning conditions are weak circuit condition.
We will show the following theorem:
\begin{theorem}\label{thm:weakobj}
  The winning coalition problem for weak Muller conditions is \PSPACE-complete.
\end{theorem}
\RB{donner un sketch de l'algorithm ici?}

\begin{proposition}\label{prop:localvisited}
  For weak Muller conditions, the value of a history only depends on the set $\states{h}$ of states that have been visited and on $s=\last(h)$ the last state of the history.
\end{proposition}

\begin{collect}{appendix-weak}{\subsubsection{Proof of Proposition~\ref{prop:localvisited}}}{}
\begin{proof}
  A weak Muller condition can be transformed into a prefix-independent one by remembering which states have already been visited, this is a consequence of the fact that for prefix-independent objectives the value depends only on the last state of the history.   
\end{proof}
\end{collect}

We will thus simply write $\val^n_j(\visited,s)$ for $\val^n_j(h)$ where $\visited=\states{h}$ and $s=\last(h)$.
The core of the argument consists in proving that the value for given $R$ and $s$ can be computed in polynomial space.

\begin{lemma}
For every iteration $n$, \player{i}, state $s$, and set of states~$\visited$, the value of $\val^n_j(\visited,s)$ can be computed in polynomial space.
\end{lemma}

\begin{proof}
We define a recursive procedure $\textsf{Proc}(\visited,n)$ which computes the value $\val^n_i(\visited,s)$ for each \player{i} and state $s$.
This procedure relies on the conditions $\Theta^n_i$ (and $\Psi^n_i$) from Section~\ref{sec:valcomput}, but require them to maintain $\visited$ as the set of states visited so far.
Namely,
\begin{mathpar}
\Theta^n_i(\visited) = \Theta^n_i \cap \prestates{\visited}
\and\textrm{and}\and
\Psi^n_i(\visited) = \Psi^n_i \cap \prestates{\visited}.
\end{mathpar}

Given $\visited$ and $n$, $\textsf{Proc}(\visited,n)$ works as follows.
\begin{itemize}
\item Compute the condition $\Theta_i^n(\visited)$.
  In that case $\win{i}$ is a trivial condition (either \emph{true} if $\visited$ satisfies $\win{i}$ or \emph{false} otherwise) so we can use the polynomial space algorithm of Section~\ref{sec:prefix-independent}.\MS{Or is it juste the (sub)Section~\ref{sec:valcomput}?}
  Note that the computation of this condition may require computing value of some state for some $n'<n$, this will be done by recursive call to $\textsf{Proc}(\visited,n')$.
\item Do the same thing for $\Psi_i^n(\visited)$.
\item For each transition $s \rightarrow s'$ with $s\in \visited$ and $s'\not\in \visited$ we compute $\val^n_i(\visited',s')$ with $\visited'=\visited\cup\{s'\}$.
  This is done by a recursive call to $\textsf{Proc}(\visited',n)$; note that $\visited'$ strictly contains $\visited$.
\item 
  If from $s$, there is a strategy of \player{i} such that any outcome either:
  \begin{itemize}
  \item stays in $\visited$ and is winning for $\Theta_i^n(\visited)$;
  \item or the first state reached outside of $\visited$ has value $1$;
  \end{itemize}
  then $s$ has value $1$.
\item 
  If there is no strategy profile that satisfies $\Phi^n_i(\visited)$ and no profile whose outcome from $s$ leaves $\visited$ for $s' \notin \visited$ with $\val^n_i(\visited\cup{s'},s')\geq0$, then $s$ has value $-1$.
\item 
  In the other cases $s$ has value $0$.
\end{itemize}

At each recursive call either $n$ decreases or the size of $\visited$ strictly increases, therefore the height of the stack of recursive call is bounded by $n + |V|$.
The computation, ignoring recursive calls, can be done in polynomial space.
Thus the global procedure works in polynomial space.
\qedhere
\end{proof}


\begin{proof}[\proofname\ of Theorem~\ref{thm:weakobj}]
Let $W$ and $L$ be the sets of players that should win and lose, respectively.
We guess a lasso path $\rho$ such that $\states{\rho} = \visited$ such that for $i \in W$, $\visited \vDash \varphi_i$ and for $i \in L$, $\visited \nvDash \varphi_i$, where for $i \in \Agt$, $\varphi_i$ is the formula over $V$ defined by the circuit for the winning condition of $i$.

Let $n_0$ be the iteration after which the values do not change, and assume we have, with $\textsf{Proc}(\visited,n_0)$, computed the values for all players and states.

We check that $\rho$ is the outcome of an admissible profile.
In order to do that, $\rho$ is divided into $\rho=\rho_1 \cdot \rho_2^\omega$.
Note that we can choose $\states{\rho_1}=\visited$, $|\rho_1| \leq |V|^2$ (since a visit to a state of $\visited$ may require a path of length at most $|V|$), and  $|\rho_2|\leq|V|$ (we have also $\states{\rho_2} \subseteq\visited$).
We first check that in $\rho_1$, no player ever lower its own value.

We then check that $\rho_2^\omega$, the looping part of $\rho$, is part of an admissible profile that stays in $\visited$, \ie:
\[\rho_2^\omega \in \outcome(\stratset^{n_0}) \cap \prestates{\visited}\]
where $\outcome(\stratset^{n_0})$ is computed as in Section~\ref{sec:valcomput} with trivial objectives for all players: \[\outcome(\stratset^{n_0}) = \bigcap_{n=0}^{n_0}\bigcap_{i \in \Agt}\L(\A^n_i).\]
In particular, no player can decrease its own value.
Note that the definition of the winning condition of these automata requires the values of states at previous iterations, which is provided by calls to $\textsf{Proc}(\visited,n)$.

\medskip
\PSPACE-hardness is given by Lemma~\ref{lem:pspace-hard}.
\qedhere
\end{proof}

In section~\ref{sec:safety}, the winning coalition problem was proved \PSPACE-hard even for the special case of safety conditions.
The problem is in fact also \PSPACE-hard for reachability conditions.
\begin{lemma}
  The winning coalition problem for reachability conditions is \PSPACE-hard, even for sets of players 
such that $|W|=1$ and $L=\emptyset$.
\end{lemma}
\begin{proof}
  This is essentially the same proof than for safety, we will therefore only insist on the differences.
  \begin{itemize}
  \item If $\phi = x_i$ we define the module $M_\phi$ in which player $x_i$ has a choice between winning or letting the game continue.
  \item If $\phi = \lnot x_i$ the construction is similar, with player $\lnot x_i$ replacing $x_i$.
  \item The other modules are defined in the same way than in Lemma~\ref{lem:pspace-hard}.
  \end{itemize}
  The game $\G_\phi$ is obtained by directing the remaining outgoing transitions to a winning state for \Eve.
  
  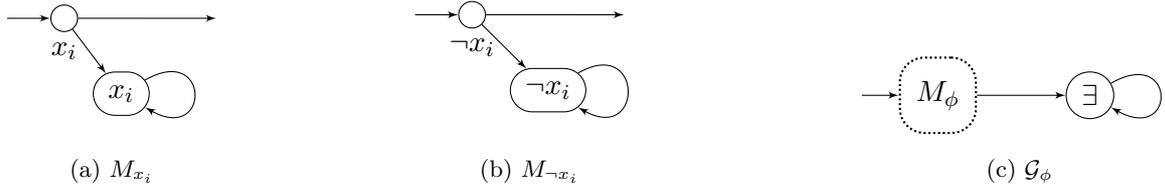
\begin{figure}
\centering
\tikzstyle{every state}+=[rounded rectangle,inner sep=5pt,minimum size=8pt]
\tikzstyle{module}+=[state,rectangle,draw,thick,densely dotted,minimum size=1cm,rounded corners=3.25mm]
\tikzstyle{every picture}+=[>=latex']
\subfloat[$M_{x_i}$]{\label{fig:MxReach}
  \begin{tikzpicture}
  \draw (2,2) node[state,label=below:$x_i$] (L11) {};

  \draw (2.75,1) node[state] (E) {${x_i}$};

  \draw[->] (1.25,2) -- (L11);
  \draw[->] (L11) edge (4,2);
  \draw[->] (L11) edge (E);
  \draw[->] (E) edge[bigloop right] (E);
  \end{tikzpicture}
}
\hfill
\subfloat[$M_{\lnot x_i}$]{\label{fig:MnxReach}
  \begin{tikzpicture}
  \draw (2,2) node[state,label=below:$\lnot x_i$] (L11) {};

  \draw (3,1) node[state] (E) {${\lnot x_i}$};

  \draw[->] (1.25,2) -- (L11);
  \draw[->] (L11) edge (4,2);
  \draw[->] (L11) edge (E);
  \draw[->] (E) edge[bigloop right] (E);
  \end{tikzpicture}
}
\hfill~
\subfloat[$\G_{\phi}$]{\label{fig:GReach}
  \begin{tikzpicture}
    \draw (3,1) node[module] (M1) {$M_{\phi}$};
    \draw (5,1) node[state] (F) {$\shortEve$};

    \draw[->] (2,1) -- (M1);
    \draw[->] (M1) edge (F);
    \draw[->] (F) edge[bigloop right] (F);
  \end{tikzpicture}
}
\caption[Modules for the definition of the game $\G_\phi$.]{Modules for the definition of the game $\G_\phi$. A label $y$ \emph{inside} a state $q$ denotes that $q\in\good_y$; a label $y$ \emph{below} a state $q$ denotes that $q\in V_y$. Note that \eve is abbreviated \shortEve.}

  \end{figure}

Given a history $h$ of the game, we write $\won(h)=\{ i\in\Agt \mid \exists k.\ h_k \in \good_i\}$ the set of players who already won in that history.
We associate to such a set of player a valuation $v_\won$ such that:
\[
v_\won(x_i) = \left\{\begin{array}{ll}
1 & \textrm{if } x_i \in \won \\
0 & \textrm{if } \lnot x_i \in \won \wedge x_i \notin \won \\
\textrm{undefined} & \textrm{otherwise} \\
\end{array}\right.
\]

If $\ell$ is a literal, the strategies of $\stratset^1_\ell$ are exactly the one such that if $\ell\not\in \won(h)$ and $\last(h)\in V_\ell$, take the transition to the state in $\good_\ell$.

Now \eve should only visit a state controlled by $\ell$ if it has already been visited in the past.
We are in the same situation than in the proof of Lemma~\ref{lem:pspace-hard}, so the remainder of the proof is exactly the same.
\qedhere
\end{proof}

%% file: conclusionFull.tex
This paper provides new fundamental complexity and algorithmic results on the iterated elimination of dominated strategies in the context of turn-based $n$-player perfect information $\omega$-regular games. We have shown that the non-elementary complexity of the plain tree-automata-based procedure suggested in~\cite{berwanger07} can be avoided, and our precise complexity results show that iterated elimination of dominated strategies is a {\em computationally viable alternative} to Nash equilibria.

Future directions include investigating extensions with imperfect information, with quantitative objectives, or richer interaction models like concurrent games.